\documentclass[11pt]{article}

\title{The Toom Interface Via Coupling}
\author{Nicholas Crawford, Wojciech De Roeck and Gady Kozma}

\usepackage{amsmath, amsfonts, amssymb, amsthm}

\usepackage{enumerate, color}
\usepackage[colorlinks=true, linkcolor=black, citecolor=black,urlcolor=blue,pdfborder={0 0 0},pdfborderstyle={}]{hyperref}
\usepackage{multido}
\usepackage{url}
\usepackage{mathrsfs}
\usepackage{mathtools}
\usepackage{verbatim}
\usepackage{setspace}
\usepackage{nicefrac,bbm,calrsfs,xspace}
\usepackage{stmaryrd}
\usepackage[clock]{ifsym} 

\usepackage{cleveref}
\usepackage{hyperref}

\crefname{theorem}{Theorem}{Theorems}
\crefname{lemma}{Lemma}{Lemmas}
\crefname{proposition}{Proposition}{Propositions}
\crefname{section}{\S}{\S\S}
\crefname{equation}{}{}

\usepackage[margin=3cm]{geometry}
\usepackage{tikz}

\newif\iffinal
\finalfalse 
\finaltrue  
\iffinal\else\usepackage[notref,notcite]{showkeys}\fi

\newtheorem{theorem}{Theorem}[section]
\newtheorem{lemma}[theorem]{Lemma}
\newtheorem{proposition}[theorem]{Proposition}
\newtheorem{definition}[theorem]{Definition}

\newtheorem{claim}[theorem]{Claim}
\theoremstyle{remark}

\newtheorem{remark}[theorem]{Remark}

\newcommand{\R}{\mathbb{R}}
\newcommand{\N}{\mathbb{N}}

\newcommand{\E}{\mathbb{E}}
\newcommand{\Z}{\mathbb{Z}}
\newcommand{\D}{\mathbf{D}}
\newcommand{\Q}{\mathbb Q}
\renewcommand{\Pr}{\mathbb{P}}
\newcommand{\bbone}{\mathbbm{1}}

\newcommand{\Var}{\operatorname{Var}}

\newcommand{\Om}{\Omega}

\DeclareMathOperator{\Ber}{Ber}

\DeclareMathOperator{\maxl}{max\mathrm{l}}
\DeclareMathOperator{\minr}{minr}
\DeclareMathOperator{\mix}{mix}
\DeclareMathOperator{\Dens}{Dens}
\newcommand{\Bp}{{\Ber_p}}
\renewcommand{\d}[1]{\nabla_{#1}}

\newcommand{\eps}{\epsilon}

\newcommand{\mf}{\mathfrak}

 \DeclareMathOperator{\supp}{Supp}

\newcommand\otimesal{\mathop{\hbox{\raise 1.6 ex
  \hbox{$\scriptscriptstyle\mathrm{al}$}
\kern -0.92 em \hbox{$\otimes$}}}}
\newcommand\oplusal{\mathop{\hbox{\raise 1.6 ex
  \hbox{$\scriptscriptstyle\mathrm{al}$}
\kern -0.92 em \hbox{$\oplus$}}}}
\newcommand\Gammal{\hbox{\raise 1.7 ex
\hbox{$\scriptscriptstyle\mathrm{al}$}\kern -0.50 em $\Gamma$}}

\let\al=\alpha   

  \let\ga=\gamma 
 \let\la=\lambda \let\om=\omega 
\let\si=\sigma

   \let\Om=\Omega

\newcommand{\caB}{{\mathcal B}}

\newcommand{\caF}{{\mathcal F}}
\newcommand{\caG}{{\mathcal G}}

\newcommand{\KK}{{\mathcal K}}

\newcommand{\caL}{{\mathcal L}}

\newcommand{\caN}{{\mathcal N}}
\newcommand{\NN}{{\mathcal N}}

\newcommand{\caP}{{\mathcal P}}
\newcommand{\caQ}{{\mathcal Q}}

\newcommand{\bbE}{{\mathbb E}}

\newcommand{\bbP}{{\mathbb P}}

\newcommand{\bbZ}{{\mathbb Z}}

\newcommand{\opunit}{\text{1}\kern-0.22em\text{l}}

\newcommand{\str}{ |}

\newcommand{\e}{{\mathrm e}}

\renewcommand{\d}{{\mathrm d}}

\newcommand{\beq}{ \begin{equation} }
\newcommand{\beqs}{ \begin{equation*} }
\newcommand{\eeqs}{ \end{equation*} }
\newcommand{\eeq}{ \end{equation} }
\newcommand{\bet}{ \begin{theorem} }
\newcommand{\eet}{ \end{theorem} }

\newcommand{\wdr}{\textcolor{blue}}

\begin{document}

\maketitle
\begin{abstract}
We consider a one dimensional interacting particle system which describes the effective interface dynamics of the two dimensional Toom model at low noise.  We prove a number of basic properties of this model.  First we consider the dynamics on a finite interval $[1, N)$ and bound the mixing time from above by $2N$.  Then we consider the model defined on the integers.  Because the interaction range of the rates and the jump sizes can be arbitrarily large, this is a non-Feller process.  As such,  we can define the process starting from product Bernoulli measures with density $p \in (0, 1)$, but not from arbitrary measures.
We show that the only possible invariant measures are those product Bernoulli measures, under a modest technical condition. We further show that the unique stationary measure on $[0, \infty)$ converges to i.i.d.\ Bernoulli variables when viewed far from 0.
\end{abstract}

\section{Introduction}\label{introduction}
In this paper, we consider an interesting interacting particle system originally introduced in \cite{DLSS} to describe the effective dynamics of the interface between two phases in Toom's\footnote{Pronounce Toom with a long o, not with the English pronunciation of oo} Model (also known as the North-East, or North-East-Center, model) in the limit of weak noise.  We recall here (see \cite{Toomy} for more details) that Toom's model is a discrete time probabilistic cellular automaton on $\Z^2$ in which the spin configurations $\sigma_t \in \{-1,1\}^{\Z^2}$ are updated in parrallel according to the rule 
\[
\sigma_{t+1}(i, j)=
\begin{cases}
\textrm{sign}\left(\sigma_{t}(i, j+1)+\sigma_{t}(i+1, j)+ \sigma_{t}(i, j)\right)  \text{ with probability $1-p-q$}\\
+1 \text{ with probability $p$}\\
-1   \text{ with probability $q$}.
\end{cases}
\]
The parameters $p, q$ represent noise in the update scheme.  It is remarkable, and important for what follows, that, for $p$ and $q$ small enough, the system has two stationary states, one with mostly $+1$'s and the other with mostly $-1$'s.

One may impose an interface between these two phases by setting the model up in the third quadrant of $\Z^2$ and fixing boundary conditions for $\sigma_{t}(j, 0)=+1$ and $\sigma_{t}(0, j)=-1$ for all $j<0$ and for all $t$.  

If $p=q=0$, all ``up-right'' paths from $(-\infty, -\infty)$ to $(0, 0)$ define stable configurations (with $+$ above and $-$ below the path) for the deterministic dynamics.  One may then ask how these interfaces fluctuate for $p, q \neq 0$ but small.
 Heuristically, one may expect that flips off the line separating the $+$ and $-$ regime die out quickly, and the dynamics of the line is governed by flips on it. For example, a spin-flip at a vertex immediate to the left of long vertical segment of the interface will generally cause further spin-flips at vertices adjacent to the segment of interface below. The net effect on the interface is to shift the part of the segment below that point 1 unit to the left. Encoding vertical edges of the line by $+1$ and horizontal by $-1$, the authors of \cite{DLSS} arrive at the following effective description of the dynamics at weak noise.  
First of all, the up-right paths are encoded by spin configurations $\sigma:=(\sigma({x}))_{x \in \N} \in \{-1, 1\}^\N$ (1 corresponds to a vertical segment of the interface and $-1$ to a horizontal segment).  Second, the dynamics on up-right paths is described by a continuous time Markov chain.  Each $1$ particle is equipped with an exponential rate $\lambda_{+}$ clock, and a $-1$ particle with a rate $\lambda_{-}$ clock. We will assume throughout that $\la_+,\la_->0$ and that $\la_++\la_-=1$, the latter condition simply fixes of unit of time. When the clock rings for a particle of fixed sign, the particle exchanges positions with the first particle to its right of opposite sign.  \emph{From now on we will refer to these dynamics as the \textit{Toom Interface}. We will not return to the two dimensional dynamics}. Here and below, we will denote this process by $\sigma_t:= (\sigma_{t}(x))_{x \in \N}$.


One interesting feature of this model is that its restriction to the first $L$ vertices $\{\pm 1\}^{\llbracket 1,L\rrbracket}$ is itself a Markov chain; the dynamics is the same unless a clock rings for a spin in the last block of constant sign in $\llbracket 1,L\rrbracket$.  For updates of spins in the last block, the dynamics reduces to single vertex spin flips. In the language of \cite{DLSS}, there are no finite-size effects.  It is easy to check that the restricted chain is irreducible on $\{\pm 1\}^{\llbracket 1,L\rrbracket}$ and hence has a unique stationary measure $\mu_L$.  The sequence of measures $(\mu_{L})_{L\in \N}$ is consistent, and this in turn implies that the full chain has a unique invariant measure on $\{\pm 1\}^\N$, $\mu_{\infty}$, which restricts to $\mu_L$ on $\{\pm 1\}^{\llbracket 1,L\rrbracket}$.

Very little is understood rigorously regarding the behavior of either $\mu_{\infty}$ or the process $\sigma_t$, though the papers \cite{BFLS,DLSS,DS} contain a number of interesting conjectures, heuristics and numerics.  The first paper on the subject, \cite{DLSS}, studied the Markov chain defined above as a model describing fluctuations via kinetic roughening, the height function $h_x(\si_t)$ being defined by $h_x(\si_t)=\sum_{i=1}^x \si_t(i)$. The striking observation there is that if $\lambda_+= 1/2$,  the statistical properties of the model cannot be in the class governed by the conventional KPZ equation: in this case the process $h_x(\si_t)$ is distributionally invariant under global spin flip. 

One way to understand this at a heuristic level is to follow the work of Kardar, Parisi and Zhang and guess the behavior of $h(\si_t)$  in the appropriate scaling limit.  The process should satisfy the SPDE
\[
\partial_t h= \kappa \Delta h  + W(t, x) +a (\nabla h)^2+ b (\nabla h)^3  \dotsc ,
\]
where $W$ is a space-time white noise and the last set of terms make explicit the possible dependence on the gradient of $h$.
If $\lambda_{+} \neq \lambda_{-}$, one concludes that only the quadratic term is relevant using scaling theory \cite{krugspohn}.  However, if $\lambda_+=\frac 12$, $h$ and $-h$ are identically distributed, which forces $a=0$ in any putative scaling limit. The extent to which the third order term is relevant is an intriguing open question.  It is marginal in the renormalization group sense, and as such \cite{DLSS, DS} argue against its  appearance for the scaling limits of microscopic models.  The situation here is analogous to the expected relationship between the scaling limit of the Ising model in $4$ dimensions and the putative $\phi^4_4$ field theory.

The simplest manifestation of the above discussion appears in the study of the variance, under $\mu_{\infty}$, of the sum of the first $L$ spins.  Numerics, Renormalization group calculations and heuristics \cite{DLSS,PBMM92} suggest that
\[
\textrm{Var}_{\pi_{\infty}}\Big(\sum_{x=1}^L \si_x\Big) \sim
\begin{cases}
L^{2/3} \text{ if $\lambda_{+} \neq \frac12$},\\
L^{1/2}\log^{1/4} L \text{ if $\lambda_{+} = \frac12$.}
\end{cases}
\]
It might help the reader to consider what this implies about the correlations in this model. If $\Var \ll L$ then the model must exhibit strong negative correlations to cancel the contribution to the variance coming from the term $\sum \si^2 =L$.

With this background in mind, our paper constitutes the first rigorous analysis of the Toom interface, though the results fall short of answering the most intriguing questions raised in \cite{DLSS}, e.g.\ the above conjecture on the variance. (The paper \cite{ToomTsetlin} analyzes a similar, but different model). Let us now present our main findings.
We recall that the total variation distance between two measures $\mu, \nu$ on a finite sample space $\Omega$ is defined as 
\[
\|\mu-\nu\|:=\frac 12 \sum_{\sigma \in \Omega} |\mu(\sigma)-\nu(\sigma)|
\]
Abusing notation slightly, we also use $\sigma_t$  to denote the restriction of the chain to $\{\pm 1\}^{\llbracket 1,L\rrbracket}$ and let $\sigma^{\xi}_{t}$ denote the distribution of $\sigma_t$ when starting from the initial configuration $\xi\in \{\pm 1\}^{\llbracket 1,L\rrbracket}$.   Recall that the mixing time of $\sigma_t$ is defined as
\[
\tau_{\mix}(L):= \inf\left\{ t: \max_{\xi} \| \sigma^{\xi}_{t}- \mu_L\| < \frac 12\right\}
\]
Our first result is as follows.
\begin{theorem}
\label{L:Mix}
For all $L \in \N$,
\[
\tau_{\mix}(L) \leq 2L
\]

\end{theorem}

A potentially surprising property of the Toom interface model is that it can be defined on the whole of $\Z$.  In this case the Bernoulli i.i.d.\ measures are invariant to the dynamics. In other words, the phenomenon of unusually small variance seems to disappear (notwithstanding that one expects, as in ASEP, to recover small variances when studying certain dynamic observables such as the current across an appropriately chosen space-time characteristic).  We wish to understand this disparity better.  We will do it in two different directions.

The first direction is to study the behavior of $\mu_{\infty}$  in the bulk, far to the right of $0$. Is it Bernoulli? Note that on $\Z$ all i.i.d. product Bernoulli measures, with any density $p$, are invariant to the dynamics. On $\N$ far from the boundary, the prospective Bernoulli measure is fixed by the condition $\E_{\mu_\infty}[\si_x]=p$ (this being dependent on $\la_+$ and $\la_-$). 

Formally, let $\tau_x$ be the translation by $x$ i.e.\ for any spin configuration $\sigma$, with domain $D\subset \Z$, let $\tau_x\sigma$ denote the spin configuration with domain $D+x$ defined by $(\tau_x\sigma)(y)= \sigma(y-x)$. Denote the induced map on the space of probability measures by $\tau_x^*$. Studying the behavior of $\mu_\infty$ to the far right is thus studying $\lim_{k\to-\infty}\tau_k^* \mu_{\infty}$.
\begin{theorem}
\label{T:WeakN}
Consider $(\tau_{k}^*\mu_{\infty})_{k \in - \N}$ as a sequence of probability measures on $\{\pm 1\}^\Z$. Then this sequence converges weakly, as $k \rightarrow -\infty$, to the i.i.d.\ Bernoulli measure $\Bp$ with 
$$
\left(\frac{1-p}{p}\right)^2=\frac{\lambda_+}{\lambda_-}
$$
\end{theorem}

The second direction is to ask: are there any measures on $\Z$ invariant to the dynamics other than the i.i.d.\ Bernoulli measures? We show that none exist, under some conditions which promise that information does not flow too fast from $-\infty$. While we failed to construct an ``exotic'' (i.e.\ non-i.i.d.) invariant measure, we have no good reason to conjecture such an example does not exist, it seems a condition on the flow really is necessary. As the specific conditions we use are somewhat lengthy to state, we defer the statement of this result, \Cref{T:Station}, to the next section.


 
\subsection{Proof ideas}\label{S:sketches}

The main tool that we employ is a coupling. Let $\si^1$ and $\si^2$ be two starting configurations. We wish to construct a coupling of the Toom processes starting from $\si^i$ which makes them attempt to become similar with time. We perform this coupling as follows. We start with independent Poisson clocks (one for each vertex) each with rate 1. Suppose there is a Poisson arrival at time $t$ and at a site $x$. We examine $\sigma_t^1(x)$ and $\si_t^2(x)$. If $\sigma_t^1(x)=\sigma_t^2(x)$ we want the particles at $x$ to move together.  To obtain the proper particle clocks we have to reduce the rate, so we throw a coin with probability $\nicefrac 12$ and make them both walk if it succeeds (for this informal discussion we assume $\la_+=\la_-=\nicefrac 12$, the $\la_+\ne\la_-$ case is similar). If $\si_t^1(x)\ne\si_t^2(x)$ then we again throw a coin with probability $\nicefrac 12$: if it succeeds we make $\si^1$ walk, and if it fails, we make $\si^2$ walk. It is easy to check that both $\si^i_t$ are Toom processes, so this is indeed a coupling.

Let us examine {\bf discrepancies} i.e.\ $x$ such that $\si^1(x)\ne\sigma^2(x)$. A Poisson arrival at $x$ will force $\si^1(x)=\si^2(x)$ after it, but a discrepancy might form somewhere to the right of $x$, call this site $y$. We say that the discrepancy at $x$ moved to $y$. Our discrepancy might also move because of a Poisson arrival before $x$, but the key point is that it in all cases it moves to the right. The important points regarding discrepancy dynamics are as follows: discrepancies are never created.  When they move, they only move to the right. They may annihilate each other, but only by collisions between opposite types; e.g. a ``$+$ discrepancy'' (a discrepancy where $\sigma^1(x)=1$ and $\sigma^2(x)=-1$) hits a ``$-$ discrepancy''.  Let us note here that this coupling is \textit{attractive}:  If $\si_0^1\leq \si^2_0$ pointwisely, then they remain so for all future time.

Let us sketch how the coupling gives our results.

\begin{proof}[Sketch of a proof of \Cref{L:Mix}]Recall that we want to show that the mixing time on a finite interval of length $L$ is bounded by $2L$. We couple two processes on this finite interval with arbitrary starting configurations and examine the discrepancies. They move right with speed bigger or equal to $\nicefrac 12$ and fall off the right edge. By time $2L$ they are all gone, and the configurations are the same. This is well-known to imply a mixing time bound.
\end{proof}

\begin{proof}[Sketch of a proof of \Cref{T:WeakN}]Recall that we wish to show that the Toom process on $\N$, examined at $x$, is approximately Bernoulli. We couple the half-line process to the full-line process with the correct $p$. In this cases the coupling may create discrepancies.  For, a  Poisson arrival at some non-positive $x$ which moves a particle to some $y\in \N$ does not have a half-line process counterpart (note that whether a discrepancy is created or not depends also on whether $\si^1(y)$ agrees with $\si^2(y)$ prior to the arrival). We wish to analyze the flow of discrepancies across a half-space deep in the bulk.  We therefore move to a version where the coupling is stationary too (we already have that both coupled processes are stationary, but the coupling is not necessarily so). This is done using a more-or-less standard limit process. We then examine the rate at which discrepancies flow past a point $x\in\N$ (denote it by $j_x$). We show that $j_x$ is a decreasing function of $x$, with $j_x-j_{x+1}$ being exactly the rate of annihilations at $x$. 
To prove the result, it is enough to show that $j_x$ tends to $0$ as $x\to\infty$.  Since the rate of annihilations must be small, the only way that $j_x$ cannot tend to $0$ is if there are, with positive density, long stretches (in space or in time) of discrepancies of a single sign. However discrepancies can only annihilate in pairs, so long stretches of discrepancies of the same sign correspond to periods of time in which the signed sum of discrepancies across $0$ is large.  Finally, we show that the latter cannot happen often enough to support a non-zero limit for $j_x$.
\end{proof}
\begin{proof}[Sketch of a proof of \Cref{T:Station}]The theorem will state that the only stationary measures on $\Z$ are Bernoulli. To show this, we start with a stationary measure $\mu$ and couple it to all Bernoulli processes at once (themselves coupled so that for every $p<q$ the Bernoulli-$q$ process is pointwise larger than the Bernoulli-$p$ process). To couple more than two Toom processes at once, just do as follows: once the site to move is selected, throw a random coin, it it falls on heads move all 1s, and if it falls on tails move all $-1$s. As in the previous proof sketch, we construct a version where the coupling itself is also stationary. We then show that there cannot be any annihilations in the coupling, as the flow of discrepancies is stationary, and annihilation would cause the set of discrepancies to  decrease with time (we need here that the flow is finite and space-bounded, which induces some conditions on our measure $\mu$). This means that, compared to any of the Bernoulli-$p$ measure coupled to it, it is either pointwise bigger than it everywhere, or pointwise smaller. There is, thus, a critical $p$ (possibly random) such that $\mu$ is Bernoulli-$p$ (perhaps except at one point). From here it is not difficult to conclude that $\mu$ is a mixture of Bernoulli measures.\end{proof}

Of the three, the most accessible is the proof of \Cref{L:Mix} appearing in \Cref{S:Apps}. 
\Cref{S:InvZ} is devoted to a proof of \Cref{T:Station} while \Cref{S:WeakN} gives a proof of \Cref{T:WeakN}.  Except for some notation set out at the beginning of \Cref{S:InvZ} these latter two sections may be read independently of one another.  \Cref{S:auxiliary} contains a number of lemmas used in both \Cref{S:InvZ} and \Cref{S:WeakN}, notably the existence of a stationary coupling. 

Let us also mention a second paper \cite{CKR2}.  In that paper, we prove various functional central limit theorems for additive functionals of local observables, local currents, tagged particles and the like.  Combining the results of that paper with the present paper, we are in fact able to derive the bound 
\[
\textrm{Var}_{\pi_{\infty}}\Big(\sum_{x=1}^L \si_x\Big) \lesssim L.
\]
Going beyond this bound probably requires a new idea beyond the technology developed here and in \cite{CKR2}.

\subsection*{Acknowledgements}
We thank Joel Lebowitz for motivating us to work on the problem and for telling us that the i.i.d.\ measure is invariant on the whole line. We thank Christian Maes for pointing out the relation with reference \cite{alexander}.  WDR acknowledges the support of the DFG (German Research Fund) and the Belgian Interuniversity Attraction Pole  P07/18 (Dygest). NC is supported by Israel Science Foundation grant number 915/12. GK is supported by Israel Science Foundation grant number 1369/15 and by the Jesselson Foundation.

\section{The Main Coupling and Dynamics on \texorpdfstring{$\{\pm1\}^\Z$}{Z}}
\label{S:coupling}

The heuristic given in the run-up to \Cref{T:WeakN} presupposes that the dynamics may actually be defined on $\Z$.  This is a nontrivial issue as
 the process does not have a finite interaction range -- arbitrarily distant parts of the configuration on the negative axis can influence the local jump rate -- and hence the standard Hille-Yosida construction, as outlined e.g.\ in \cite{LiggettBook}, is not applicable.  As far as we know, only a few non-Feller interacting particle systems  have been constructed, most of them relying on a monotonicity property that is missing here, see e.g.\ \cite{maesredig, Liggettnonfeller}. But beyond applicability of standard tools, there are serious issues of existence and uniqueness. The process is not defined starting from arbitrary starting conditions in any reasonable sense -- how would one go about defining it if the starting conditions are, say, $+$ on the entire negative line? Further, even if for a given starting configuration and collection of Poisson arrivals there exists a version of the process which is defined for all time, it is not clear that such a version is necessarily unique. For example, suppose the starting configuration is $\dotsb++--++--\dotsb$ and that there is a Poisson arrival at $-2n$ at time $1/n$, for all $n$ (with $-2n$ the first in the block of two signs). Then forcing the spin at $-4n$ to jump 2 units and all those at $-4n+2$ to jump one unit is a legal solution, but so is its opposite. We have no example of an initial configuration where such non-uniqueness occurs with positive probability and constructing such an example should be interesting. \label{pg:nonuniqueness} To circumvent this we start by requiring uniqueness, which is encoded in \Cref{def:Tproc} below. But first some preliminaries.

Rather than thinking of $\lambda_{\pm}$-Poisson clocks as being attached to particles, we will consider a sequence of i.i.d.\ rate one Poisson point processes $(N_x(t))_{x\in \Z}$ associated with vertices $x\in \Z$.  Besides these Poisson point processes, we need a two dimensional array of of i.i.d.\ uniform $[0,1]$ variables $(U_{x, j})_{x\in \Z, j\in \N}$. Let $\Omega$ be a probability space realizing all these variables. 

\label{pg:Omega}There are a number of  ways to realize a probability space supporting these variables.  Because we want to use time shifts to construct Toom interfaces, we shall specify one concrete setting.  The state space on which the Poisson processes are defined will be $\Om_1= D([0,\infty)\to \N^\Z)$, the space of c\'adl\'ag functions from $[0,\infty)$ to $\N^\Z$ where $\N^\Z$ is equipped with the usual product topology. We equip $\Om_1$ with the Skorokhod topology and sigma algebras and the probability measure is defined to be the product distribution such that for each $x$, the random variable $\pi_x(\om_1):=\om_1(t, x)$ is distributed as the aforementioned Poisson process $N_x(t)_{x\in \Z}$.   

Let $\Omega_2=[0,1]^{ \Z\times \N}$ with its usual product topology and sigma algebras.
Finally let  $(\varOmega, \mathbb P; \caB_{\varOmega})$ denote the probability space obtained by taking the Cartesian product of these two probability spaces.  Abusing notation, we will from now on denote the coordinate projections associated with $\Om_1$ (respectively $\Omega_2$) by $(N_x(t))_{x\in \Z}$, respectively  $(U_{j, x})_{ j\in \N, x\in \Z}$.

Finally, let $D= D([0,\infty)\to\{\pm 1\}^\Z)$ be the space of c\'adl\'ag functions from $[0,\infty)$ to $\{\pm 1\}^\Z$ where $\{\pm 1\}^\Z$ is equipped with the usual product topology. We equip $D$ with the Skorokhod topology.
\begin{definition}\label{def:Tproc}
A T-process is a pair $(\mu,F)$ where $\mu$ is a probability measure on $\{\pm 1\}^\Z$ and $F$ is a function from $\{\pm 1\}^\Z\times\varOmega$ to $D$ which is Borel measurable. $F$ need only be defined $\mu\times\mathbb{P}$-almost everywhere. We require $F(\eta,\omega)(0)=\eta$ and further that $\{F(\eta,\omega)(t):t\in[0,T]\}$ is measurable with respect to $\caF_T$ where $\caF_T$ is the natural time filtration on $\{\pm 1\}^\Z\times \Omega$ (defined formally below).

Alternatively, a T-process is a $D$-valued random variable $\sigma$ such that for some $(\mu,F)$ as above, $\Pr(\sigma\in E)=(\mu\times\mathbb P)(F^{-1}(E))$.

We denote $F_t(\eta,\omega)=F(\eta,\omega)(t)$ and similarly $\si_t$ is the random variable on $\{\pm 1\}^\Z$ given by $\si$ at time $t$.

\end{definition}

The natural time filtration $\caF_t$ are $\sigma$-algebras on $\{\pm 1\}^\Z\times\Omega$ defined by
\[
\caF_t = \sigma\left(\eta;\: N_x(s): s \leq t; \: U_{x, k}: k \leq N_x(t)\right)
\]
where $\eta$ is the first coordinate (the element of $\{\pm 1\}^\Z$). Note that we think about $U_{x,j}$ as associated with the $j^\textrm{th}$ jump of the Poisson process at $x$, which is reflected in the definition of $\caF_t$.

Nothing has yet been formalized regarding Toom processes in this definition, so we just called it a ``T-process'' in anticipation of the Toom model, which enters in the next definition. We remark that $\si$ clearly determines $\mu$, as $\si_0$ is distributed according to $\mu$, and  determines $F$ $\mu\times\mathbb P$-almost everywhere, which is enough, since $F$ is anyway defined only $\mu\times\mathbb P$-almost everywhere.

\begin{definition}\label{def:Toomproc}For $\sigma\in D$, $\omega\in\varOmega$, $t\in[0,\infty)$ and $x<y$ in $\Z$, we say that a Toom update happens for $\si$ at $(t,x,y)$ if the following occurs:
\begin{enumerate}
\item There is a Poisson arrival at $(t,x)$ i.e.\ $N(x,t)=N(x,t^-)+1$.
\item If $\sigma_x(t^-)=1$ then $U_{x,N_x(t)}$ is required to be less than $\la_+$, otherwise it is required to be bigger than $\la_+$.
\item $\sigma_x(t^-)=\sigma_{x+1}(t^-)=\dotsb=\sigma_{y-1}(t^-)=-\sigma_y(t^-)$.
\end{enumerate}
A Toom interface (which, abusing names, we often call a Toom process) on $\Z$ is a T-process $(\mu,F)$ such that $\mu\times\mathbb P$-almost surely, $F(\eta,\omega)$ has the following properties
\begin{enumerate}[a)]
\item $F_t(\eta,\omega)(x)$, considered as a function of $t$, has only finitely many jumps in any finite interval, for any $x$. 
\item If a Toom update happens for $F(\eta,\omega)$ at $(t,x,y)$, then there are spin flips at $x$ and $y$ at time $t$. Otherwise there are no jumps at time $t$.
\end{enumerate}
If $S$ is a finite or a semi-infinite interval, we define a Toom process on $S$ in the same manner, except that in Clause 2 we require $x\in S$.  Also, we shall say there is a Toom update at $x$ if $x$ is the left endpoint of a Toom update at $(t, x, y)$.
\end{definition}

We are now in a position to define our coupling, which is simply using the same $\omega\in \varOmega$ to run a number of different Toom processes. Formally,
\begin{definition}\label{def:coupling}Let $\{(\mu^i,F^i):i\in I\}$ be two or more Toom processes (not necessarily on the same subset of $\Z$). When we discuss a ``coupling of the $(\mu^i,F^i)$ started from $\mu$'' we mean the following: $\mu$ is assumed to be a measure on $\prod_{i\in I}\{\pm 1\}^\Z$ whose marginals are the $\mu^i$. The coupling is then the collection of the $D$-valued random variables $\si^i$ given by
\[
\sigma^i=F^i(\eta^i,\omega)\qquad\sigma^i:\Big(\prod_{i\in I}\{\pm 1\}^\Z\Big)\times\varOmega\to D.
\]

The coupling with independent starting positions is the object given when $\mu$ is taken to be $\prod \mu_i$.
\end{definition}

This coupling has a number of nice features.  First and foremost, it is \textit{attractive}.  To see what we mean by that, let us introduce the partial order on spin configurations $\sigma^1 \geq\sigma^2$ if $\sigma^1(x) \geq\sigma^2(x)$ for all $x\in S$. Let $(\mu_i, F_i)_{i=1}^2$ be a pair of Toom processes on $S$ coupled together and started from a measure $\mu$ satisfying $\mu(\si^1\geq \si^2=1)$. When we say that the coupling is attractive we mean that in this case
\[
\mathbb P(F_1(t)\geq F_2(t) \text{ for all $t \in \R_+$})=1.
\]
To see this one has to simply check the various cases.  For example, if $\sigma^1\geq \sigma_2$ and $\sigma^2(x)=+1$, a Toom update is required in both if $U_{x,N_x(t)}\leq \lambda_+$.  If $z_i$ is the first minus to the right of $x$ in $\sigma^i$, then $\sigma^1\geq \sigma_2$ implies $z_1\geq z_2$.  This then implies the ordering must be preserved by the update.  The remaining cases are left to the reader to check.

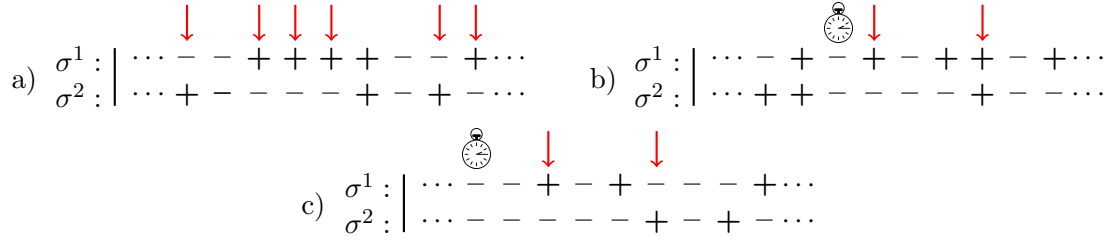
\begin{figure}
\centering
\begin{tikzpicture}[xscale=.48, yscale=.48]

\node at  (-8.5, .4)  {a)};

\node at (-7, 1) {$\sigma^1:$};
\node at (-7, 0) {$\sigma^2:$};
\draw[-][thick](-6,-.3)--(-6,1.3);

\node at (-5, 1) {$\cdots$};
\node at (-4, 1) {\textbf{--}};
\node at (-3,1) {\textbf{--}};
\node at (-2,1) {\textbf{+}}; 
\node at (-1,1) {\textbf{+}};
\node at (0,1) {\textbf{+}}; 
\node at (1,1) {\textbf{+}}; 
\node at (2,1) {\textbf{--}}; 
\node at (3,1) {\textbf{--}}; 
\node at (4,1) {\textbf{+}};
\node at (5, 1) {$\cdots$};

\node at (-4, 0) {\textbf{+}};
\node at (1,0) {\textbf{+}}; 
\node at (0,0) {\textbf{--}}; 
\node at (2,0) {\textbf{--}}; 
\node at (3,0) {\textbf{+}}; 
 \node at (-1,0) {\textbf{--}};
 \node at (-2,0) {\textbf{--}}; 
 \node at (-3,0) {\textbf{--}};
 \node at (-3,0) {\textbf{--}};
  \node at (4,0) {\textbf{--}};
 \node at (5, 0) {$\cdots$};
  \node at (-5, 0) {$\cdots$};

 \draw[->,thick, red] (-4, 2.5)--(-4, 1.5);
  \draw[->,thick, red] (-2, 2.5)--(-2, 1.5);
   \draw[->,thick, red] (0, 2.5)--(0, 1.5);
    \draw[->,thick, red] (3, 2.5)--(3, 1.5);
        \draw[->,thick, red] (-1, 2.5)--(-1, 1.5);
        \draw[->,thick, red] (4, 2.5)--(4, 1.5);
\end{tikzpicture}
\hfill
\begin{tikzpicture}[xscale=.48, yscale=.48]

\node at  (-8.5, .4)  {b)};

\node at (-7, 1) {$\sigma^1:$};
\node at (-7, 0) {$\sigma^2:$};
\draw[-][thick](-6,-.3)--(-6,1.3);

\node at (-5, 1) {$\cdots$};
\node at (-4, 1) {\textbf{--}};
\node at (-3,1) {\textbf{+}};
\node at (-2,1) {\textbf{--}}; 
\node at (-1,1) {\textbf{+}};
\node at (0,1) {\textbf{--}}; 
\node at (1,1) {\textbf{+}}; 
\node at (2,1) {\textbf{+}}; 
\node at (3,1) {\textbf{--}}; 
\node at (4,1) {\textbf{+}};
\node at (5, 1) {$\cdots$};

\node at (-5, 0) {$\cdots$};
\node at (-4, 0) {\textbf{+}};
\node at (-3,0) {\textbf{+}};
\node at (-2,0) {\textbf{--}}; 
\node at (-1,0) {\textbf{--}};
\node at (0,0) {\textbf{--}}; 
\node at (1,0) {\textbf{--}}; 
\node at (2,0) {\textbf{+}}; 
\node at (3,0) {\textbf{--}}; 
\node at (4,0) {\textbf{--}};
\node at (5, 0) {$\cdots$};

\node at (-2,2) {\Taschenuhr};
  \draw[->,thick, red] (-1, 2.5)--(-1, 1.5);
   \draw[->,thick, red] (2, 2.5)--(2, 1.5);
\end{tikzpicture}
\begin{tikzpicture}[xscale=.48, yscale=.48]

\node at  (-8.5, .4) {c)};

\node at (-7, 1) {$\sigma^1:$};
\node at (-7, 0) {$\sigma^2:$};
\draw[-][thick](-6,-.3)--(-6,1.3);

\node at (-5, 1) {$\cdots$};
\node at (-4, 1) {\textbf{--}};
\node at (-3,1) {\textbf{--}};
\node at (-2,1) {\textbf{+}}; 
\node at (-1,1) {\textbf{--}};
\node at (0,1) {\textbf{+}}; 
\node at (1,1) {\textbf{--}}; 
\node at (2,1) {\textbf{--}}; 
\node at (3,1) {\textbf{--}}; 
\node at (4,1) {\textbf{+}};
\node at (5, 1) {$\cdots$};

\node at (-5, 0) {$\cdots$};
\node at (-4, 0) {\textbf{--}};
\node at (-3,0) {\textbf{--}};
\node at (-2,0) {\textbf{--}}; 
\node at (-1,0) {\textbf{--}};
\node at (0,0) {\textbf{--}}; 
\node at (1,0) {\textbf{+}}; 
\node at (2,0) {\textbf{--}}; 
\node at (3,0) {\textbf{+}}; 
\node at (4,0) {\textbf{--}};
\node at (5, 0) {$\cdots$};

\node at (-4,2) {\Taschenuhr};
 \draw[->,thick, red] (1, 2.5)--(1, 1.5);
 \draw[->,thick, red] (-2, 2.5)--(-2, 1.5);
\end{tikzpicture}
\caption{a) Discrepancy locations; b) An update induces non-local discrepancy motion (the clock marks the point of update); c) An update induces non-local annihilation.}\label{DFig}
\end{figure}

Next, we want to highlight the quasi-particles of this coupling, which we call discrepancies.  These are the Toom interface analogs of second-class particles from the study of exclusion processes.
Let us define them formally now.  Given two spin configurations $\sigma^1, \sigma^2$, let 
\begin{align}
\label{eq:DisC}
&\D^\eta= \D^\eta_\sigma:= \{x \in \Z: \sigma^1(x) =\eta, \: \sigma^2(x)=-\eta\}\qquad\eta=\pm 1,\\[1mm]
&\D= \D_\sigma:= \D^{+}_\si\cup \D^{-}_\si = \{x \in \Z: \sigma^1(x) \neq \sigma^2(x)\}
\end{align}
If two Toom interfaces are coupled together, we can interpret the evolution of discrepancy locations as the evolution of a collection of particles.  If $\sigma^1_{t-}(x) \neq \sigma^2_{t-}(x)$ and there is a Toom update at $x$, then necessarily $\sigma^1_{t}(x) = \sigma^2_{t}(x)$.  In this case we view the discrepancy as having moved from $x$.  A discrepancy at a point $x$ can move due to a Toom update at $x$, or due to a Toom update at some vertex $z<x$. When it moves, there are two possible effects which may occur elsewhere in the configurations; either a new discrepancy may appear at a vertex $y$ or another discrepancy located at a position $y$ may disappear (see Figure \ref{DFig}) .  The important point is that in both cases $y>x$.  In other words discrepancies \emph{always move to the right!}  This observation is at the heart of everything we do in this paper.
Finally, let us note that there are multiple ways of viewing discrepancy motion --- either the discrepancy that was at $x$ ``jumps over" neighboring discrepancies to move to $y$ or the discrepancy at $x$ collides with a discrepancy to its right, at say $w$, taking its place at $w$ and causing the discrepancy that was at $w$ to move.  There is then a chain reaction of discrepancy collisions until a last discrepancy moves into $y$.

\subsection{Statement of Theorem \ref{T:Station}}\label{S:statestat}

\Cref{T:Station} is about stationary Toom processes, so we should start by defining those, but before we need to define the natural time shifts. Recall from page \pageref{pg:Omega} that $\Omega=\Omega_1\times\Omega_2$ and the definition of the $\Omega_i$.
For $\om\in \varOmega$, $\omega=(\omega_1, \om_2)$ with $\om_i\in \Om_i$, let $S_t\om=(\om_{1, x}(\cdot+t)-\om_{1, x}(t), \om_{2, y}(\cdot+N_t(y)))_{x\in \Z, y\in \Z}$.

On the space $\{\pm 1\}^\Z\times\varOmega$ with its usual product topology and sigma algebra $\caF\otimes \caB_{\varOmega}$, define the filtration of sigma algebras $(\caF_t)_{t \in \R^+}$  by

\begin{definition}\label{def:muFstat}
A T-process $(\mu,F)$ is stationary if 
\begin{enumerate}
\item $F$ preserves $\mu$, i.e.\ for any $t\in(0,\infty)$ and any $E\subset\{\pm 1\}^\Z$ Borel, 
\[
\int \bbone_E(F_t(\eta,\omega))\,d\mu(\eta)\,d\mathbb P(\omega)=\mu(E).
\]
\item $F$ forms a semi group i.e.\ $F_s(F_t(\eta,\omega),S_t\omega)=F_{s+t}(\eta,\omega)$ for all $t,s\in(0,\infty)$, $\mu\times\mathbb P$-almost everywhere. 
\end{enumerate}
It is called ``stationary on $S$'' for some interval $S\subset\Z$ if we only require clause 1 to hold for events $E$ which depend only on $S$, formally if $\eta\in E, \eta|_S=\eta'|_S\Rightarrow \eta'\in E$.
\end{definition}
We need one more technical condition. 

\begin{definition}\label{def:regular}
A T-process $(\mu,F)$ is called regular if for every $t$ one can write $F_t$ as the $\mu\times\mathbb P$-limit in measure of functions $F_t^L$ such that $F_t^L(\cdot,\omega)$ is continuous (as a function from $\{\pm 1\}^\Z$ to itself), for almost all $\omega$.
\end{definition}

We have no example of a Toom process which is not regular. Heuristically, constructing a non-regular example seems a similar challenge to constructing an example of a non-unique Toom process (recall the discussion on page \pageref{pg:nonuniqueness}). We will not do it here, but it is possible to formulate very mild conditions of ``no flow of information from infinity at finite time'' which would ensure that a process is regular. We are not very happy about this condition, but at least, as will be shown in \Cref{sec: construction n}, it is very easy to check in concrete cases. It will be used only once, in the proof of the next lemma.

\begin{lemma}\label{lem:coupling exists}
Let $(\mu^1,F^1)$ and $(\mu^2,F^2)$ be two stationary, regular Toom processes on $S^1$ and $S^2$ respectively. Let $\nu$ be any measure on $\{\pm 1\}^{S_1\cup S_2}$ which has $\mu^1$ and $\mu^2$ as its marginals. Let $\nu_t$ be the result of applying the coupling to $\nu$ for time $t$, i.e.
\[
\nu_t(E)=\int \bbone_E((F^1(\eta^1,\omega)(t),F^2(\eta^2,\omega)(t))\,
d\nu(\eta^1,\eta^2)\,d\mathbb P(\omega).
\]
Then any subsequential weak$^*$-limit of $\frac 1T\int_0^T\nu_t$ is a stationary coupling of $(\mu^1,F^1)$ and $(\mu^2,F^2)$.
\end{lemma}
(weak$^*$ convergence here is in the functional analytic sense --- what is sometimes called in probability weak convergence).
As this lemma is technical in nature we postpone its proof to \Cref{sec: invar}. The lemma will be used in  \Cref{S:InvZ,,S:WeakN}, in \Cref{S:InvZ} for $S_1=S_2=\Z$ and in \Cref{S:WeakN} with $S_1=\Z$ and $S_2=\N$. 

Finally, let ${{l_y}}(\sigma)$ and $ {{r_y}}(\sigma) $ denote the  cardinality of the maximal block of spins with the same sign to the left of $y$, starting from $y-1$ and respectively to the right of $y$, starting from $y+1$, in particular ${{l_y}}, r_y \geq 1$. 

\begin{theorem}
\label{T:Station}
Let $\si$ be a stationary, regular Toom process on $\Z$. 
If the  integrability condition 
\item  \[
\sup_{x \in \Z} \E[(l_x)^{1+\eps}(\si_0)]< \infty,
\] 
holds {for some $\epsilon>0$}, then $\si_0$ is distributed as a mixture of product Bernoulli measures.
\end{theorem}
By making additional assumptions on $\si$, e.g. spatial translation invariance, one may assume weaker moment conditions.  However to keep the presentation streamlined, we will stick with the above hypotheses on all measures encountered below.

\subsection{Construction of standard processes}  \label{sec: construction n}

First, let us construct a Toom process $(\mu,F^L)$ on a finite interval $S$. For concreteness, we choose $S=\llbracket 1,L\rrbracket=\{1,\dotsc,L\}$. We restrict  the measure $\mu$ to give full weight to configurations $\sigma_0=\eta$ having both infinitely many $+$'s and infinitely many $-$'s in $\N$. We define ${F}^L_t$ as follows. When there is a Poisson arrival at $x \in S$ at time $t$ and if $U_{x, N_x(t)}$ and $\sigma_x(t^-)$ satisfy Condition 2) of \eqref{def:Toomproc}, then look for the smallest $y>x$ such that $\sigma_{t_-}(y)=-\sigma_{t_-}(x)$ and flip both $\sigma(x)$ and $\sigma(y)$.
The condition that $\sigma_0$ has infinitely many $+$'s and $-$'s in $\N$ ensures that for the first Poisson arrival in $\llbracket 1,L\rrbracket$, we can find such a $y$, and this will remain true after finitely many jumps.
 Thus $ F^L$ is well-defined unless there are two Poisson arrivals at the same time or infinitely many Poisson arrivals in a finite interval of time. Since both have probability 0, our $ F^L$ is defined $\mathbb P$-almost everywhere, which, as already mentioned, is enough.  Our description of $(\mu, F^L)$ obviously matches the formal definition of a Toom process on $S$. The process is regular because it is in fact itself continuous for all $\omega$ for which it is defined.
  
 It is important to note that the configuration $\sigma(x), x>L$ plays no role, except for providing a reservoir of $\pm$ spins. In particular, if the restrictions of $\eta, \eta'$ to $I_L$ agree, $\eta(I_L)=\eta'(I_L)$, and both $\eta,\eta'$ have infinitely many $\pm$ in $\N$, then, almost surely for all $t$, 
 $$
 {F}^L_t(\eta,\omega)(x)= {F}^L_t(\eta',\omega)(x) \text{ for all $x\in \llbracket 1,L\rrbracket$}.
 $$
 We use this observation to define the function $\widetilde F^L_t(\cdot,\omega): \{\pm\}^{\llbracket 1,L\rrbracket} \to \{\pm\}^{\llbracket 1,L\rrbracket} $ so that it coincides with $ F^L(\eta, \omega)(\llbracket 1,L\rrbracket)$, almost surely. The only distinction between $\widetilde F^L_t$ and $ F^L_t$ is that in the former we omit spins, and spin flips,  at $x \notin \llbracket 1,L\rrbracket$.     By construction then, $\widetilde F^L_t$ is a Markov chain on $\{\pm\}^{\llbracket 1,L\rrbracket}$ and we now argue that it is irreducible.  To get from a configuration $\eta\in \{-1, +1\}^{\llbracket 1,L\rrbracket}$ to another, $\eta'$, first make the Poisson clocks of all $-$ sites in $\eta$ ring from left to right, getting to the all $+$ configuration. Then have all $-$ sites in $\eta'$ ring from left to right, getting to $\eta'$ Hence, it follows that $\widetilde F^L_t$ has a unique invariant measure $\mu^L$, which we will need below. Considering $\mu^L$ as a measure on $\{\pm 1\}^\Z$ (which ignores $x\notin\llbracket 1,L\rrbracket$) we get that $\mu^L$ is invariant to $F^L_t$ and is the unique such measure.
  
 Next we note that the $ F^L$ are consistent in the sense that for $x\in \llbracket 1,L\rrbracket$, $F^L_t(\eta,\omega)(x)=F^M_t(\eta,\omega)(x)$, for all $M>L$, all $t$, and almost all $\eta$ and $\omega$. 
Hence the limit (in $\{\pm 1\}^{\Z}$)
$$
F_t(\eta,\omega)=\lim_{L\to \infty}  F^L_t(\eta,\omega)
$$ 
exists almost surely.  These observations give us a regular Toom process $(\mu,F)$ on $\N$.

Next, we construct a stationary Toom process on $\N$.  The invariant measures $\mu^L$ constructed above, are consistent (due to uniqueness) and hence they define a measure $\mu^\infty$ on $\{\pm 1\}^{\N}$. To get a stationary Toom process we want to choose the starting measure $\mu$ in the above construction such that its restriction to $\N$ coincides with $\mu^\infty$. However, for this to be legitimate, we need to verify
\begin{lemma}$\mu^\infty$ gives zero measure to configurations with a tail of a unique sign.\end{lemma}
\begin{proof}Fix $L$ and let $M\gg N$. Examine the event in $\mu_M$ that $\si(L)=\si(L+1)=\dotsb=\si(M)$ and assume for concreteness that the common value is $+$. The process exits this state with rate at least $\lambda_+(M-L)$, as any Poisson arrival in $\{L,\dotsc,M\}$ with an appropriate $U$ will exit this state. However, it is not difficult to check that returning to this state requires at least one Poisson arrival in $\{1,\dotsc,L-1\}$ or in $M$. So its rate is bounded by $L$. We get that the probability of this event is bounded above by $L/(\lambda_+(M-L))$. Taking $M\to\infty$ shows that the probability of $\si(L)=\si({L+1})=\dotsb$ in $\mu_\infty$ is zero. As $L$ was arbitrary, the lemma is proved.\end{proof} 

It is easy to conclude from the stationarity of $(\mu^L,F^L)$ that $(\mu^\infty,F^\infty)$ is stationary. Furthermore,
the Toom process just described is regular because it is a limit of $F^L_t\to F_t$ also $\mu^\infty\times\mathbb P$ almost surely. 

We can immediately provide some payoff for the work done so far by proving \Cref{L:Mix}.
\begin{proof}[Proof of \Cref{L:Mix}]
Fix $\phi\in\{\pm 1\}^{\llbracket 1,L\rrbracket}$ to be arbitrary, and let $\psi\in\{\pm 1\}^{\llbracket 1,L\rrbracket}$ be distributed according to the stationary measure $\mu^L$. Start two Toom processes on $\llbracket 1,L\rrbracket$ from $\phi$ and $\psi$ and couple them as above (call the resulting processes $\sigma^\phi_t$ and $\sigma^\psi_t$ respectively). If we show that at some time $T$ that $\Pr(\sigma^\phi_T=\sigma^\psi_T)>\frac12$ then $T$ upper-bounds the mixing time by definition.

As explained in the introduction, a crucial property of our coupling is that it pushes discrepancies to the right.  Let us formalize this statement. Let $\tau_1$ be the first arrival on site $1$ and for all $j \geq 2$ let $\tau_j$ be the first arrival on site $j$ after $\tau_{j-1}$.  Because we are looking at the process with a wall to the left of site $1$, once $\tau_1$ occurs, the value of $\sigma^{\phi}_t(1)=\sigma^{\psi}(1)$ for all $t \ge \tau_1$.  By induction, the same is true for all $\{(j, \tau_j): j \leq L\}$.
The theorem is proved by observing that $ \tau_L$ is the time it takes for the $L$'th arrival of a Poisson point process which has rate $1$.
\end{proof}

Let us remark that there is also a coupling-less version of the argument. Indeed, examining only one process, after $\tau_1$ the value of $\sigma_t(1)$ is independent of the initial configuration, and similarly for all $\tau_j$. Thus $\tau_N$ is a \emph{forget time} and this is equivalent to the mixing time, see \cite{LW}.

The last result we wish to show here is that i.i.d.\ Bernoulli-$p$ processes have a corresponding $F$ which makes them into a stationary, regular Toom process on $\Z$. This is a known folk result, but it seems not to have appeared in writing so we put it here for completeness.

\begin{lemma}\label{lem:Ber exists}There is an $F:\{\pm 1\}^\Z\times\varOmega\to D$ such that for all $p\in(0,1)$ the couple $(\Ber_p,F)$ is a stationary, regular Toom process.\end{lemma}

Let us isolate the first step of the proof as a separate claim.
\begin{lemma}To prove  \Cref{lem:Ber exists} it is enough to construct such an $F$ which has the required properties only for $t<\epsilon$ for some fixed $\epsilon$.\end{lemma}
\begin{proof}
Exchange $\epsilon$ and $2\epsilon$, and call the input of the lemma $G$, i.e.\ $G:\{\pm 1\}^\Z\times\varOmega\to D$ and it has the required properties (i.e.\ from definitions \ref{def:Tproc}, \ref{def:Toomproc} and especially from \ref{def:muFstat} and \ref{def:regular})  only for $t<2\epsilon$ ($t+s<2\epsilon$ for the property that it forms a semigroup). We form $F$ by repeatedly applying $G$ i.e.
\begin{align*}
F_t(\eta,\omega)&=G_t(\eta,\omega)&t&\le\epsilon\\
F_t(\eta,\omega)&=F_{t-\eps}(F_\eps(\eta,\omega),S_\epsilon\omega)&t&>\epsilon
\end{align*}
where $S_\epsilon$ is the time shift on $\varOmega$, as in the previous section. Note that we are using here stationarity: $F_{t-\eps}$ is defined only $\mu\times\mathbb{P}$-almost everywhere, so the expression only makes sense because the couple $(F_\eps(\eta,\omega),S_\eps\omega)$ has $\mu\times\mathbb P$ as its law.

It is easy to check that $F$ preserves $\mu$ and that it is a Toom process on $\Z$. To check that $F$ forms a semigroup, assume first that $t<\epsilon$ and get
\begin{align*}
F_s(F_t(\eta,\omega),S_t\omega)&=
F_{s-\eps}(\underbrace{F_\eps(F_t(\eta,\omega),S_t\omega)},S_{\epsilon+t}\omega)\\
&=F_{s-\eps}(F_{t+\eps}(\eta,\omega),S_{t+\epsilon}\omega)\\
&=F_{s-\eps}(F_t(F_\eps(\eta,\omega),S_\epsilon\omega),S_{t+\epsilon}\omega)\\
&=F_{s+t-\eps}(F_\eps(\eta,\omega),S_\epsilon\omega)\\
&=F_{s+t}(\eta,\omega).
\end{align*}
The first equality follows from opening the outer $F$ by its inductive definition, the second is obtained by applying the semigroup property to the inner term (marked by a brace), which is allowed since $t+\epsilon<2\epsilon$. The third is reopening in the opposite order of $t$ and $\epsilon$. The fourth is by assuming the semigroup property has been proved inductively for $s-\epsilon$ and the fifth is again the definition of $F$. This shows the case $t<\epsilon$ by induction on $s$. Concluding the case of general $t$ is similar and we will skip it.

Finally we need to show that $F$ is regular. We show that by induction on $t$ so we will assume it has already been proved for $t$ and will demonstrate it up to $t+\eps$ (our assumption on $G$ is the induction base). 
In other words, our assumption is that for any $\delta>0$ one may find a function $F^L_{t}$  as in definition \ref{def:regular}, i.e., continuous in its first variable almost surely in its second variable, such that 
\[
\mu\times\mathbb P\Big(\Big\{(\eta,\omega):d(F_t(\eta,\omega),F_t^L(\eta,\omega))<\frac 12\delta\Big\}\Big)>1-\frac13\delta.
\]
Since $F^L_t(\cdot,\omega)$ is defined on a compact space, it has a modulus of continuity (which might depend on $\omega$). Take such a modulus which holds for all but $\frac13\delta$ probability, i.e.\ a $\gamma$ which satisfies
\[
\mathbb P\Big(\Big\{\omega:d(\eta,\eta')<\gamma\Rightarrow d(F^L_t(\eta,\omega),F^L_t(\eta',\omega))<\frac12\delta\Big\}\Big)>1-\frac 13\delta.
\]
Finally use the regularity of $G$ to pick a $G^M_\eps$ such that 
\[
\mu\times\mathbb P(\{(\eta,\omega):d(G^M_\eps(\eta,\omega),G_\eps(\eta,\omega))<\gamma\})>1-\frac13\delta.
\]
Thus $F^L_t(G^M_\eps(\eta,\omega),S_\epsilon\omega)$ is a $\delta$-approximation of $F_t(G_\eps(\eta,\omega),S_\eps\omega)$, which, by the definition of $F$, is the same as $F_{t+\epsilon}$. And of course, it is continuous for almost all $\omega$. 
This finishes the lemma.
\end{proof}
\begin{proof}[Proof of \Cref{lem:Ber exists}]
We will construct $F$ by taking the limit of finite systems with periodic boundary conditions. Let us define the system with periodic boundary conditions formally, even though there are no surprises.
We define $\rho=\rho_{t}^{(L)}$ to be the following process on $\{\pm1\}^{(-L,L]}$
given as a function of Poisson arrivals: suppose we have an arrival
at time $t$ and position $x$. Let $y$ be the cyclically first point
right of $x$ having opposite sign, i.e.\ $\rho_{t-}(y\mbox{ mod }2L)=-\rho_{t-}(x)$
and $\rho_{t-}(z\mbox{ mod }2L)=\rho_{t-}(x)$ for all $x<z<y$ (here
and below, $y$ mod $2L$ is the element of $(-L,L]$ congruent to
$y$ modulo $2L$). Now define 
\[
\rho_{t}(z)=\begin{cases}
\rho_{t-}(y) & z=x\\
\rho_{t-}(x) & z=y\\
\rho_{t-}(z) & \mbox{otherwise.}
\end{cases}
\]
(for completeness let us stipulate that if $\rho_{0}$ is the configuration
with all $+$ or all $-$ then $\rho_{t}=\rho_{0}$ for all $t$).
A simple check shows that for every $k\in\{0,\dotsc,2L\}$, the measure
which is uniform over configurations with exactly $k$ $+$ signs
is stationary. Hence so are the Bernoulli-$p$ measures. Denote by $F^L$ the map $\{\pm1\}^\Z\times\varOmega\to D$ which realizes this process on $(-L,L]$ and freezes the configuration outside the interval.

We will now show that $F^{2^k}_t(\eta,\omega)$ converges $\Bp\times\mathbb P$-almost surely for all $p$ as $k\to\infty$. This will construct $F$, show that it is regular and that it preserves $\Bp$. By the previous lemma, it is enough to show this claim only up to some small fixed time $\epsilon$. We will choose $\epsilon$ later.

Compare therefore $F^{L}(\eta,\omega)$ and $F^{2L}(\eta,\omega)$ restricted to some small spatial interval, say $[-K,K]$. The following is a sufficient  (if far from necessary) condition for them to be equal on $[-K,K]$: there is an $x\in (-L,-K)$ such that no particle passed over $x$ in either $F^{L}$ or in $F^{2L}$ in the time interval $[0, \eps]$. If we show such $x$ exist, the lemma will be proved. We call such $x$ \textit{regeneration points}. 

It will be convenient at this point to switch to discrete time. Let therefore $t_1<t_2<\dotsb$ be the Poisson arrivals in $(-L,L]$ arranged in increasing order, and denote by $x_i$ the position of the $i^\textrm{th}$ arrival. Clearly,
\[
\Pr(t_{4\eps L}<\eps)< 1-e^{-c\eps L}.
\]
From here on we call such inequalities ``with exponentially large probability'' (the $c$ in the exponent will be allowed to depend on $\la_+$ and $p$).

Examine the following parameter
\[
I_k:=\sum_{i=1-L}^{-K-1}r_i^2(\rho_{t_k})
\]
(recall that $r_i$ is the size of the block of spins to the right of $i$, set to $2L$ if $\rho$ has all $+$ or all $-$).
Since $\rho_t$ is Bernoulli, we see that $I_t\le CL$ with exponentially large probability. Denote
\[
\caG_k=\bbone\{I_\ell<CL\;\forall \ell\le k, t_k<\eps\}.
\]
and get that $\Pr(\caG_k)>1-e^{-c\eps L}$ for all $k<4\eps L$.

Next examine the sum of the lengths of the Toom updates in the interval $(-L,-K)$, i.e.
\[
A_k:=r_{x_k}(\rho_{t_k})\bbone\{x_k\in (-L,-K),\caG_k\}\qquad A:=\sum_{k=1}^{4\eps L}A_k.
\]
We now prove that, if $L\geq 2K$, 
\begin{equation}\label{eq: tail a}
\Pr(A>C_1\epsilon L)\le \frac{C}{\epsilon L}.
\end{equation}
To show \eqref{eq: tail a}, define
$$B_k:=A_k-\E[A_k|A_1, \dotsc, A_{k-1}]\qquad B:=\sum_{k=1}^{4\eps L}B_k.$$
On the one hand,
\begin{equation} \label{eq: b bound}
\E[B^2] \leq 4 \sum_k \E[A_k^2]  \leq  C\epsilon L
\end{equation}
On the other hand,
\begin{align*}
\E[A_k|A_1, \dotsc, A_{k-1}] &=
 \frac{C}{L-K-2} \E\left[I_{k}\bbone\{\caG_k\}\, |\,A_1, \dotsc, A_{k-1}\right]\le C
\end{align*}
where the equality follows since the \textit{location} of the $k^\textrm{th}$ arrival is independent of $A_1, \dotsc,\linebreak[0] A_{k-1}$; and the inequality follows since $I_k\bbone\{\caG_k\}<CL$ deterministically.
Summing this up to $4\eps L$ and estmating $B$ with \eqref{eq: b bound} and Markov's inequality gives \eqref{eq: tail a}.

We are interested in $A$ because it bounds the number of non-regeneration points (for $F^{L}$ with $L=2^k$): Since the same estimate holds also for $L={2^{k+1}}$, we may fix $\epsilon$ and get that with probability larger than $1-C2^{-k}$ a regeneration point for the pair $F^{2^k}, F^{2^{k+1}}$ may be found. By Borel-Cantelli this means that there exists an $k_0$ such that for all $k>k_0$ a regeneration point exists for each pair $F^{2^k}, F^{2^{k+1}}$. As discussed above, this proves that $F$ is a regular Toom process that preserves $\Bp$.

To finish the lemma we need to show the semigroup property for $F$, up to time $\epsilon$. Fix some $K$ sufficiently large such that $F^{2^k}$ has a regeneration point in $[K,0)$ with probability larger than $1-\delta$ independently of $2^k$ (as long as $2^k>K$). Since $F^{2^k}\to F$ we see that $F$ has a regeneration point in $[K,0)$ with probability larger than $1-\delta$. Since $\delta$ was arbitrary, we get that $F$ has infinitely many regeneration points. But this shows the claim because after a regeneration point we can calculate $F$ as if it were a finite Toom process. Since this satisfies the semigroup property, so does $F$, and the lemma is proved.
\end{proof}

\subsection{A simple application}\label{S:Apps}

In this section we give a simple, yet interesting application of the coupling. Here is the precise statement:
\begin{theorem}\label{thm: correlation decay}
Consider a Toom process on $\Z$ started from $\Bp$.  Let $f,g$ be local functions with $\Bp(f)=\Bp(g)=0$ and $\Bp(f^4)=\Bp(g^4)=1$. Then
\beq
\bbE_{\Bp}(f(\si_0)g(\si_t)) \leq C  \e^{ r-ct }, 
\eeq
with  $r$ being the length of the smallest interval containing both $\supp f$ and $ \supp g$, and $C,c$ only dependent on $\la_{\pm}$. 
\end{theorem}
From now on, constants $C,c$ throughout the paper will be allowed to depend on $p,\la_{\pm}$ without further mention.   

\begin{proof}[Proof of \Cref{thm: correlation decay}] 
Take two local functions $f,g$ and, for concreteness,  say that $ \supp f \cup \supp g  \subset [0,r]$. We define the measure $\nu$ on $\si=(\si^1,\si^2) \in \left(\{\pm 1\}^\Z\right)^2$  as follows:
\begin{enumerate}
\item Both $\si^1$ and $\si^2$ are $\Bp$-distributed.  
\item For $x <0$,   $\sigma^1(x)=\sigma^2(x)$.
\item $\{\sigma^1(x): x\geq 0\}$ is independent of $\si^2$ and idem with $1 \leftrightarrow 2$. 
\end{enumerate}
Let $\si_t=(\si^1_t, \si^2_t) $ be the coupling of Toom processes started from $\nu$ (recall \Cref{def:coupling}).

Let $X(\sigma_t)$ denote the position of the ``left-most discrepancy" of the configuration $\sigma_t$, i.e.\ $X(\sigma_t):= \min \{x: \si_t^1(x) \neq \si_t^2(x) \} $. 
We are interested in $X(\sigma_t)$ for the following reason.  Since  the support of $f$ lies in $[0, \infty)$,  $f(\sigma^1_0)$ is independent of $\sigma^2_0$ and therefore of $\sigma^2_t\:\: \forall t \in \R_+$.   Moreover, if $X(\sigma_t)>r$ then $g(\sigma^1_t)= g(\sigma^2_t)$.  Since $f(\sigma^1_0)$ and $g(\sigma^2_t)$ are independent, we have
\begin{align*}
\bbE_{\nu}(f(\si^1_0)g(\si^1_t))  &=
\bbE_\nu(f(\sigma_0^1)g(\boxed{\sigma_t^1})\bbone\{X(\sigma_t)>r\})+
\bbE_\nu(f(\sigma_0^1)g(\sigma_t^1)\bbone\{X(\sigma_t)\le r\})\\
&= \bbE_\nu(f(\sigma_0^1)g(\boxed{\sigma_t^2})\bbone\{X(\sigma_t)>r\})+
\bbE_\nu(f(\sigma_0^1)g(\sigma_t^1)\bbone\{X(\sigma_t)\le r\})\\
&= \bbE_\nu(f(\sigma_0^1)g({\sigma_t^2}))-
\bbE_\nu(f(\sigma_0^1)g({\sigma_t^2})\bbone\{X(\sigma_t)\le r\}) \\
&\qquad + \bbE_\nu(f(\sigma_0^1)g(\sigma_t^1)\bbone\{X(\sigma_t)\le r\}).
\end{align*}
The first term is in the final equality is simply 0 (by the independence explained above), so we get
\begin{multline*}
|\bbE_{\nu}(f(\si^1_0)g(\si^1_t))|\le 
2\bbE_\Bp\left(f\left(\sigma_0^1\right)^4\right)^{1/4}
\bbE_\Bp\left(g\left(\sigma_0^1\right)^4\right)^{1/4}\Pr(X(\sigma_t)\le r)^{1/2}\\
\le 2\Pr(X(\sigma_t)\le r)^{1/2}
\end{multline*}
where we used Cauchy-Schwarz twice for each term, the invariance of $\Bp$ and finally that $\Bp(f^4)=\Bp(g^4)=1$.

Now, as we have already remarked a number of times, discrepancies move to the right, so $X$ must as well.  In fact, it does so at (at least) linear speed.  We see this as follows.  $X(\sigma_s)$
is naturally coupled to a  Poisson process $N(s)$ which has rate $\min (\la_{\pm})$ so that
\[
X(\sigma_t)\geq N(t)
\]
The theorem now follows from a large deviation estimate on the Poisson process $N(t)$.
\end{proof}





\subsection{Generators, Local Rates and Derivatives}\label{sec: diff coupled}
\def\conda{Condition \hyperlink{cond:A}{A}\xspace}
One unfortunate consequence of the fact that the Toom interface is non-Feller is that intuitive reasoning involving generators and locally defined rates needs to be checked rigorously. Here we state two results which justify our pervasive use of this language throughout the text, as we will almost always assume \conda below.   Because we feel these statements are more of technical rather than actual interest, their proofs are relegated to the end of the paper (see \Cref{sec: diff coupled1})

We will define our operators on the space of continuous functions from $\{\pm 1\}^S$, $S\subset\Z$, to $\R$.
We start with the flip operator $F_x$ at site $x$, defined by
\[
F_xf(\si)=f(\si^x)\qquad\sigma^x(y)=\begin{cases}
\sigma(y)&y\ne x\\
-\sigma(y)& y=x.
\end{cases}
\]
 For a finite subset $S\subset \bbZ$ and $\tilde\si \in \{-1,1\}^S$,  we have the indicators 
 $$\chi^{\tilde\sigma}_S =  \chi[{\si(x) =\tilde\si(x)} \, \forall x \in S]  $$
and whenever $\tilde \sigma$ is all $1$ or all $-1$, then we simply write $\chi^{+}_S$ and $\chi^{-}_S$. 
 We also need the associated projections --
$$P^{\tilde\sigma}_S  f(\sigma) =   \chi^{\tilde\sigma}_S(\sigma) f(\sigma).   $$
Then the generator of the process is formally defined as
 $$
 \caL =   \sum_{x < y}\caL_{x,y}   
=   \sum_{x < y}   (\la_+ P^+_{[x,y-1]}P^{-}_{y}  + \la_- P^{-}_{[x,y-1]}P^{+}_{y}  ) (F_x F_y-1)  ,
 $$
We call this definition formal because $\caL f$ is in general not continuous, due to the infinite sum over $x$.   We need a condition on moments that will almost always be assumed in the sequel:\\[1mm]
 \noindent \hypertarget{cond:A}{{\bf Condition A.}} We say that a Toom process $\sigma$ satisfies Condition A if for some $\eta>0$ 
$$
\sup_{t,x} \mathbb{E}(l_{x}^{1+\eta}(\sigma_{t}))<\infty.
$$
We say that it satisfies the local condition A if
$$
\sup_{t} \mathbb{E}(l_{x}^{1+\eta}(\sigma_{t}))<\infty\qquad\forall x.
$$
If you are reading the online version and ever forget what is \conda, clicking the letter A should send you to the definition.
%
%
\begin{lemma} \label{lem:quenched generator}   Let $\sigma$ be a Toom process satisfying the local \conda.
Let $f$ be a local function. Then for almost all $\sigma_0$ we have that 
$t\mapsto\mathbb{E}(f(\sigma_{t})\,|\,\sigma_{0})$ is differentiable
in $t$, the sum defining $(\mathcal{L}f)(\sigma_0)$ converges, and
\begin{equation}\label{eq:quenched generator}
\frac{d}{dt}\mathbb{E}(f(\sigma_{t})\,|\,\sigma_{0})=(\mathcal{L}f)(\sigma_{0})
\end{equation}
Further, we have an averaged version,
\begin{equation}\label{eq:annealed generator}
\frac{d}{dt}\mathbb{E}(f(\sigma_{t}))=\mathbb{E}\big((\mathcal{L}f)(\sigma_{0})\big).
\end{equation}
\end{lemma}
Lemma \ref{lem:quenched generator} will be proved in \S \ref{sec: diff coupled1}.
The same result also holds for a coupling of finitely many Toom processes $\sigma^{i}$: if all of the $\sigma^{i}$ satisfy the local \conda, then the differentiability of local functions also holds for the joint process, with a corresponding formal generator.  The proof of this claim follows the proof of \Cref{lem:quenched generator} in \S\ref{sec: diff coupled1} verbatim. 

Throughout the paper we will also use statements that are slight generalizations of the above. The most prominent example comes in \Cref{def:ja} and \Cref{lem:ja finite} where we handle ``instantaneous rates'' that should be justified similarly to the rates of change of $f(\sigma_t)$ above. We will henceforth use reasoning involving rates and currents without explicit justification.

\section{Invariant Measures on \texorpdfstring{$\Z$}{Z}}
\label{S:InvZ}
In this section we investigate invariant measures for the Toom interface on $\Z$, using the coupling from \cref{S:coupling}.  Throughout the section we assume that $\sigma=(\sigma^1,\sigma^2)$ is a pair of Toom processes on $\Z$ coupled as in  \Cref{def:coupling}.

Recall the definition of discrepencies \eqref{eq:DisC}.
Rather than focusing on the discrepancies themselves, it is useful to restrict attention to the study  of gaps between consecutive discrepancies of type $(+, -)$ and type $(-, +)$, i.e.\ elements of $\D^+$ and $\D^{-}$, respectively.   Let us, arbitrarily, refer to the first type discrepancies as having signature/sign $+$ and the second type of discrepancies as having signature $-$.
To keep track of ``interfaces" between the two types of discrepancy, let, for  $x \in \D^{\eta}$, 
\begin{align*}
b(x)= \inf\{y>x: y \in \D^{-\eta}\} 
\end{align*}
should such a $y$ exist and set $b(x)= \infty$ otherwise.  The set of interface discrepancies is then
\[
B=B(\sigma):=\{x\in  \D: b(x)< \infty\text{ and } (x,b(x))\cap \D=\emptyset\}.
\]

The following lemma is a main ingredient of \Cref{T:Station}, and the place where the integrability condition is used. 
\begin{lemma} 
\label{L:Dominant}
Assume $(\sigma^1,\sigma^2)$ is a stationary coupling of Toom processes which satisfies 
\[
\sup_x\E\big(\max\{l_x(\si^1),l_x(\si^2)\}\min\{r_x(\si^1)r_x(\si^2)\}\big)^{1+\eps}<\infty.
\]
Then 
\[
\Pr(\sigma^1 \leq \sigma^2 \text{ or } \sigma^2 \leq \sigma^1)=1   
\]
\end{lemma}
Here and below, ``$\le$'' stands for point-wise inequality for all $t\in[0,\infty)$ and all $x\in\Z$. We stress that the phrase ``$(\si^1,\si^2)$ is a stationary coupling'' means not only that each one is a stationary Toom process, but also that the coupling is stationary. Formally, let $\nu$ be the measure on $\left(\{\pm 1\}^\Z\right)^2$ and $F^i:\{\pm 1\}^\Z\to D$ be the maps that define $\si$ i.e
\[
\Pr(\sigma\in E)=\int \mathbbm{1}_E\big(F^1(\eta^1,\omega),F^2(\eta^2,\omega)\big)\,
d\nu(\eta^1,\eta^2)\,d\mathbb P(\omega).
\]
Stationarity here refers to the requirement that $(F^1_t(\eta^1,\omega),F^2_t(\eta^2,\omega))$ is distributed according to $\nu$ for all $t$.

In order to prove Lemma \ref{L:Dominant}, we need a number of preparatory arguments.  Some of these will also be used later in \S \ref{S:Bulk}.
Denote for brevity 
\[
\maxl_x=\max\{l_x(\si^1),l_x(\si^2)\} \qquad \minr_x=\min\{r_x(\si^1),r_x(\si^2)\}.
\] 

In general, our goal is to show that  $ \E\left[\str B  \str \right] =0$  for $\nu$ the measure on $\left(\{\pm 1\}^\Z\right)^2$ induced by $\sigma_0$. 
To orient the reader toward the direction we are headed, let us first provide an informal sketch of a proof in case $\nu$ is translation invariant. 
We fix an interval $I$.  By stationarity, we are tempted to write
\begin{equation}
\label{E:Warmup}
0=  \partial_t \E_\nu [ \str B\cap I \str ]  \leq \E[\maxl_{\min I}]- \min(\la_{\pm}) \E[|B^1 \cap I|],  \qquad   \   
\end{equation}
where $B^1 \subset B$ is the set of interface discrepancies that can be annihilated in one Toom update.  The inequality in \Cref{E:Warmup} follows from the observations that  the first term bounds from above the flow rate of of discrepancies from $(-\infty,\min I-1]$ into $I$ while the second term bounds from below the annihilation rate inside $I$. 
By hypothesis, $\sup_x \E[\maxl_x]<\infty$ so that these two inequalities together imply  that $\E[|B^1 \cap I|]$ is uniformly bounded in $I $.  
Under the assumption of translation invariance of $\nu$, this implies $\E[|B^1 |]=0$.   This argument can then be iterated (considering the discrepancies that can be promoted into  $B^1 $ in one step, etc.\ ).  Eventually one then concludes that  $\E[|B|]=0$.  Note that this argument is slightly formal due to the fact that$ \str B \cap I \str $ is not a local function and therefore the inequality in \Cref{E:Warmup} would need additional justification. Instead of remedying this directly, we proceed to  \Cref{L:Discrep1}, which we need for the general proof of  \Cref{L:Dominant} and which immediately settles the translation invariant case. We will need the following definition together with a justifying lemma, here and in \S\ref{S:WeakN}.
\begin{definition}\label{def:ja}Let $\sigma^1$ and $\sigma^2$ be two coupled Toom processes, and let $x\in\mathbb{Z}$. We define $j_x(\sigma_0)$ to be the infinitesimal rate at which a discrepancy jumps from $(-\infty,x)$ to $[x,\infty)$ (we include also in $j_x$ also the case that the discrepancy annihilates in $[x,\infty)$ in the same step). Formally, for $y< x$, $i\in\{1,2\}$ and $t>0$ we let $J_{y,t,i}$ be the event that $\sigma^i$ had a Toom update at time $t$ at position $y$, and that 1) for all $z\in[y,x)$ $\sigma^i_{t^-} (z)= \sigma^i_{t^-} (y)$ and 2) for some $z\in[y,x)$, $\sigma_{t^-}^1(z)\ne\sigma_{t^-}^2(z)$. Then we define
\[
j_x(\sigma_0)=\lim_{\eps\to 0}\frac{1}{\eps}\mathbb{P}(\exists y\le x, t\le\eps \mbox{ and } i\in\{1,2\} \mbox{ such that }J_{y,t,i}).
\]
We define $a_x(\sigma_0)$ as the infinitesimal rate of annihilation of discrepancies at $x$. Formally, 
\[
a_x(\sigma_0)=\lim_{\eps\to 0}\frac{1}{\eps}\mathbb{P}(\exists y\le x, t\le\eps \mbox{ and $i\in\{1,2\}$ s.t. $J_{y,t,i}$ and $\sigma_{t^-}^i(y)=\sigma_{t^-}^{3-i}(x)\ne\sigma_{t^-}^i(x)$}).
\]
Finally define
\[
j_x=\mathbb{E}(j_x(\sigma_0))\qquad a_x=\mathbb{E}(a_x(\sigma_0)).
\]
\end{definition}
These rates are  well-defined, as we state now.

\begin{lemma}\label{lem:ja finite} If $\si^1$ and $\si^2$ are two Toom processes satisfying the local \conda, then for any coupling of $\si^1$ and $\si^2$, $j_x(\sigma)$ and $a_x(\sigma)$ are well-defined almost surely and their averages  $a_x,j_x$ are finite.
Further, $j_x$ and $a_x$ (which have been defined above as the expected quenched rates) are equal to the corresponding annealed rates, i.e.
\begin{align*}
j_x&=\lim_{\eps\to 0}\frac 1\eps \mathbb{E}(\#\{\textrm{discrepancy jumps from $(-\infty,x)$ to $[x,\infty)$ before time $\eps$}\})\\
a_x&=\lim_{\eps\to 0}\frac 1\eps \mathbb{E}(\#\{\textrm{annihilations at $x$ before time $\eps$}\}).
\end{align*}
\end{lemma}
\begin{proof} These claims are variations on  \Cref{lem:quenched generator} as $a_x(\sigma),j_x(\sigma)$ are defined as conditional time-derivatives. The only difference with \Cref{lem:quenched generator} is that these derivatives are not naturally equal to $\E_\sigma(\caL' f)$ with $\caL'$ the formal generator of the coupled process, but the proof of \Cref{lem:quenched generator} (given below in \S\ref{sec: diff coupled1}) applies here as well.
For concreteness, let us give the expression for $a_x(\sigma)$
\[
a_y(\sigma) = \sum_{x<y}\sum_{\eta=\pm} \sum_{i=1}^2\lambda_\eta  \mathbbm 1\{\sigma\in V_{x,y,i,\eta}\}
\]
where $V_{x,y,i,\eta}$ is the event that $\eta=\sigma^i(x)=\sigma^i(x+1)=\dotsb=\sigma^i(y-1)=-\sigma^i(y)$ while $\sigma^{3-i}$ satisfies that $\sigma^{3-i}(y)=\eta$ and for some $z\in[x,y)$ we have $\sigma^{3-i}(z)=-\eta$.
\end{proof}

We now pick up the threads of the proof of \Cref{L:Dominant}, starting with
\begin{lemma}
\label{L:Discrep1}
Let $\nu$ be the initial measure of a stationary coupling and assume \conda. Then
\begin{equation}
\limsup_{ \str I \str \rightarrow \infty} \frac{1}{\str I \str} \E\left[  \str B \cap I \str  \right]=0.
\end{equation}
\end{lemma}

\begin{proof}
Given a discrete interval $[x, y]$, define the event
\[
E_{x,y}(\si):=\{x \in B(\si),\:\: y= b(x)\}
\]
In words, $E_{x, y}$ denotes the event that there is a boundary discrepancy at $x$ and that the first discrepancy to its right occurs at $y$. Recall that $a_y$ is the rate of annihilations at $y$ and that $j_x$ is the rate of discrepancy flow through $x$ defined above. We now claim that 
%
%
\begin{equation}\label{eq:ja}
j_{x}-j_{x+1} =a_x.
\end{equation}
Indeed, if we denote by $H_x(t)$ the number of discrepancies that passed from $(-\infty, x]$ to $(x, \infty)$ in the time-interval $[0,t]$, and $A_x(t)$ the number of annihilations at $x$ in $[0,t]$, then we have
$$
H_{x+1}(t)-H_x(t)=A_x(t) + \mathbbm 1(x \in \D(\sigma_0))-\mathbbm 1(x \in \D(\sigma_t))
$$
Taking the expectation, we can drop the two rightmost terms by stationarity. Dividing then by $t$ and taking $t\to 0$, we obtain \Cref{eq:ja} by using \Cref{L:Discrep1}, more precisely the identification of $a_x$ as ``annealed rates''.
 
In particular, since the $j_x$ are finite and the $a_x$ nonnegative, we have $\sum a_x\le C$.

%
%

Then \Cref{L:Discrep1} is a consequence of the following claim.

\begin{lemma}\label{eq: in the eye of the hunter} 
There is a function $c:\N \rightarrow \R$, strictly positive for all $k \in \N$ such that with $\nu$ as in \Cref{L:Discrep1}, and for any $x$ and $y$,
\beq
\label{E:Con-A}
\sum_{z =x}^y a_z  \geq   c(y-x) \nu(E_{x,y}).
\eeq
\end{lemma}
Let us finish the proof of \Cref{L:Discrep1} and then attend to \Cref{eq: in the eye of the hunter}.  
Let $I=[y,z]$ be our interval, and let $k$ be some parameter. Let $d(k)=\min_{j=1}^k c(j)$.  We get
from \Cref{E:Con-A} 
\begin{align}
\sum_{x=y}^{z} \nu(x \in B, \: b(x)-x \leq k)
&=\sum_{x=y}^z \sum_{i=1}^k\nu(E_{x,x+i})\nonumber\\
&\leq \sum_{x=y}^z \sum_{i=1}^k\sum_{j=0}^i\frac{a_{x+j}}{c(i)}
\leq \sum_{x=y}^{z+k} \frac{k^2}{d(k)}a_x
\le C(k).\label{eq:M1}
\end{align}
On the other hand, for any $k$,
\[
\sum_{x=y}^{z} \nu(x \in B, \: b(x)-x \geq k) \leq \frac{|I|+k}{k}.
\]
 \Cref{L:Discrep1} thus follows.
\end{proof}

\begin{proof}[Proof of  \Cref{eq: in the eye of the hunter}]
The main observations we make are that Toom updates due to arrivals at sites $z>y$ do not harm us --- they cannot move the discrepancy at $y$ --- while arrivals at sites $z<x$ only help us --- they can push the discrepancy at $x$ closer to or on top of $y$ (such that an annihilation occurs), or can annihilate the discrepancy where it stands, but cannot push it beyond $y$.   To arrive at an annihilation event, it suffices to have at least $y-x$ Poisson arrivals at the location of the discrepancy which starts at $x$ at time $0$ before any occur between the location of the discrepancy and $y$. 

To use this we examine the behavior in the time interval $[0,2]$. We get
\begin{align*}
\nu(E_{x,y})&=\mathbb{E}\bigg(\int_0^1\mathbbm{1}_{E_{x,y}}(\sigma_t)\,\d t\bigg)\\
&\le\mathbb{P}(\exists t\in[0,1]\mbox{ s.t. }\sigma_t\in E_{x,y})\\
&\le C(y-x)\mathbb{E}(\#\{\mbox{annihilations in $[x,y]$ before time $2$}\})\\
&=2C(y-x) \sum_{z=x}^y a_z
\end{align*}
where the last equality follows from the fact that $a_z$ is also the annealed rate of annihilations (see the ``further'' clause of \Cref{lem:ja finite}). This shows \Cref{eq: in the eye of the hunter} and consequently also \Cref{L:Discrep1}.
\comment{

To use this we need to examine the long-time behavior. We claim that
\begin{equation}\label{eq:oh no not more of this crap}
\lim_{t\to\infty}\frac{1}{t}\Big(\#\{\mbox{the number of annihilations at $x$ in $[0,t]$}\}
-\int_0^t=a_x.
\end{equation}
This is a simple exercise with Poisson processes, see the details in \Cref{lem:trivial but took me a week} below.
To see \cref{eq:oh no not more of this crap} ****starting from here its crap:**** fix some $y<x$ and examine only annihilations that come from Poisson arrivals at $y$. Let $a_{y,x}(\sigma_t)$ be the corresponding infinitesimal rate. Since $a_{y,x}(\sigma_t)$ is bounded by the number of Poisson arrivals at $y$, and since $\sigma_t$ depends only on Poisson arrivals in $[0,t]$, we easily get that
\[
\lim_{t\to\infty}\frac{1}{t}\Big(\#\{\mbox{annihilations at $x$ from arrivals in $y$ at times $[0,t]$}\}-\int_0^t a_{y,x}(\sigma_s)\,ds\Big)=0.
\]
Using stationarity and Fubini's theorem gives
\[
\lim_{t\to\infty}\frac{1}{t}\mathbb{E}(\#\{\mbox{annihilations at $x$ from arrivals in $y$ at times $[0,t]$}\})=a_{y,x}.
\]
Further, since 
\[
\int_s^ta_{y,x}=\lim_{\eps\to 0}\sum_{i=0}^{\lfloor(t-s)/\eps\rfloor}a_{y,x}(\sigma_0)\to
\]
***until here***}

\end{proof}

\begin{proof}[Proof of \Cref{L:Dominant}]
We go back to the proof of \Cref{L:Dominant}. 
  Let 
$$
B_{x}=
\{y \in B: \, b(y) \leq   x\}, 
$$
and note that $\mathbbm 1\{y \in B_x\}$ is a local function (its support is $[y,x]$) in contrast to $\mathbbm 1\{y \in B\}$, so we can use \Cref{lem:quenched generator}.   In what follows we fix some $x \in \Z$ and some $h,k \in \N$, and we omit all three from the notation to avoid clutter. Set 
\[
\theta(y)= \begin{cases}
1  &\text{ for $y\in [x-h, x]$},\\
0  &\text{ for $y \in (-\infty, -x-h-k]$},\\
1-\frac{j}{k} & \text{ for $y =-x-h-j, 0 \leq j \leq k$}.
\end{cases}
\]
By stationarity
\[
0=\partial_t \E\Big[  { \sum_{y=x-h-k}^x \theta(y) \mathbbm 1\{y \in B_x\}  }  \Big] 
\]
 Let us examine the events which change the value of the function between $\big[\cdot \big]$ in time. The value is decreased by annihilations and by discrepancies leaving the set $B_x$ (an annihilation might make the discrepancy just before it become a boundary discrepancy, i.e.\ an element of $B_x$, but since $\theta$ is increasing on $(-\infty,x]$ the sum still decreases). The latter happens, for a discrepancy at $y$, when $b(y)$ leaves $(\infty,x]$. We ignore annihilations and define 
\[
X : =\min(\lambda_\pm)\mathbbm 1\{ \exists y \in B_x\cap [x-h+1, x]: b(y)\text{ can leave $(-\infty, x]$ in one step} \}.
\]
Next we examine events which cause $\sum\theta(y)\mathbbm{1}\{y\in B_x\}$ to increase. Examine one boundary discrepancy $y$. If $r_y(\si^1)\ne r_y(\si^2)$ then, because $y$ is a boundary discrepancy, this means $y$ is separated from $b(y)$ by a stretch of equal signs. Thus, a Poisson arrival at $y$ either annihilates with the one of opposite type at $b(y)$ or moves one step to the right. Similarly, a Poisson arrival in the interval of length $[y-\maxl_y,y)$ might cause the discrepancy at $y$ to move 1 step to the right or to be annihilated. The sum in these cases increases by no more than $\theta(y+1)-\theta(y)$. In the case that $r_y(\si^1)=r_y(\si^2)$ the discrepancy at $y$ might move by this common value. Therefore, in all cases we may bound the increase in the sum by $\theta(y+\minr_y)-\theta(y)$. Defining
\[
Z(y) :=[\theta(y+ \minr_y)-\theta(y)] \,  \maxl_y \,  \mathbbm 1\{y \in B_x\},
\]
a convenient way to express the above bounds is to say
\begin{equation}
\label{E:Flow}
0\le \sum_{y\le x-h}\E_\nu[Z(y)]-\E_\nu X.
\end{equation}
(note that we use here the existence of the annealed generator, recall \Cref{eq:annealed generator}). 
To exploit \Cref{E:Flow}, we bound $\sum_y \E\left[Z(y) \right]$ by splitting the sum over $y$. 
For $x-h-k \leq y \leq x-h$, we use  $\str \theta(y+ \minr_y)-\theta(y)\str \leq \minr_y/k$ and H\"{o}lder's inequality (for any $1< p< \infty$  with $ \frac 1p + \frac 1q=1$) to get 
\[
 \sum_{y=x-h-k}^{x-h} \E\big[ \str Z(y)  \str \big] \leq      
\bigg(\frac{1}{k} \sum_{y=x-h-k}^{x-h}  \E_\nu [(\maxl_y \minr_y)^{q}]\bigg)^{\frac 1{q}}  
\bigg(\frac{1}{k}\sum_{y=x-h-k}^{x-h}  \nu (y \in B_x)\bigg)^{\frac 1{p}}
\]
If we choose $q$ sufficiently close to $1$, then the first factor is bounded by $C$ (uniformly in $h,k,x$) by the integrability assumption we placed on $\nu$.  The second factor decays as $k \to \infty$ (uniformly in $h,x$) by  \Cref{L:Discrep1} and the fact that $B_x \cap I \subset B \cap I$.   For  $ y < x-h-k$, we write
\[
\sum_{y<x-h-k}Z(y)\le\frac 1k\maxl_{x-h-k}\minr_{x-h-k}
\]
since there can be no more than one boundary discrepancy for which $\theta(y+\minr_y)-\theta(y)>0$. We conclude that
$$
\limsup_{k \to \infty} \sum_{y \le x} \E\left[Z(y) \right]  =0. 
$$
Combining with \Cref{E:Flow}, we conclude that $\E_{\nu }\left[X\right]=0$ (recall that $X$ is nonnegative and independent of $k$).
 Since this holds for all $x,h $, it follows that $\E [\str  B \str ]=0$.  
\end{proof}

One final lemma is needed before we start with the proof of \Cref{T:Station}, this time unrelated to any coupling.
\begin{lemma}\label{lem:density preserved}Let $\sigma$ be a Toom process on $\Z$ (not necessary stationary) satisfying \conda. Then with probability 1, if at some time $t_0$, the limit
\[
\lim_{L\to\infty}\frac{1}{2L+1}|\{x\in[-L,L]:\sigma_{t_0}(x)=1\}|
\]
exists (``$\sigma_0$ has a density''), then for any $t>t_0$ we also have that $\sigma_t$ has a density, and the densities are equal.
\end{lemma}
\begin{proof}Fix $T$ and $t_0<T$ and examine the quantity
\[
D_L:=\big||\{x\in[-L,L]:\sigma_T(x)=1\}|-|\{x\in[-L,L]:\si_{t_0}(x)=1\}|\big|.
\]
Then $D_L$ is the difference between the flows of particles with sign 1 into and out of $[-L,L]$, and since the rates of these flows are bounded by the lengths of the left stretches at $-L$ and $L$ we have
\[
\mathbb{E}D_L\le \E\int_{t_0}^T l_{-L}(\sigma_s)+l_L(\si_s)\,\d s\le C
\]
(where $C$ depends on the Toom process and on $T$, but not on $L$ and where the fact that the time derivative is bounded by $l$ is again by the annealed generator, see \Cref{eq:annealed generator}). 

This shows that $\frac{1}{2L+1}D_L$ converges in $L^1$ to 0 for all $t_0<T$. To pass to almost everywhere convergence we note that if we restrict, say, to squares $L=M^2$, then the Markov inequality  and the  Borel-Cantelli lemma together give us that 
\[
\frac{1}{2M^2+1}D_{M^2}\to 0\qquad\textrm{almost surely }\forall t<T.
\]
But if $\frac{1}{2L+1}|\{x\in[-L,L]:\si_t(x)=1\}|$ converges as $L\to\infty$ on the squares, then since $D_j\leq D_k$ if $j< k$, the full sequence converges with no restriction on the $L$. Finally, we note that since $T$ was arbitrary the claim just proved also holds for all integer $T$ simultaneously, proving that $\frac{1}{2L+1}D_L\to0$ for all $t\in[t_0,\infty)$.
\end{proof}
We have gathered all necessary ingredients, and we may now start the 
\begin{proof}[Proof of \Cref{T:Station}]
Let us fix a stationary measure $\mu$ for the Toom interface as in the statement of \Cref{T:Station}.  The idea is to construct a coupling $\mathbb P$ of $(\sigma^1, (\sigma(p))_{p\in [0, 1]})$ such that:
\begin{enumerate}[(a)]
\item For $p \in [0, 1]$, the distribution of $\si(p)$ is $\Bp$. 
\item The distribution of $\si^1$ is $\mu$.
\item For any $p$: ${\bf P}(\si^1 \leq \si(p) \text{ or } \si^1 \geq \si(p) ) =1$.
\item If $p'>p$, then $\si(p') \geq \si(p)$ a.s. 
\end{enumerate}

If we then define the random variable
\beq \label{def: mystery1}
\Theta:= \sup\{p \in \Q : \sigma^1 \geq \sigma(p)\} =   \sup\{p : \sigma^1 \geq \sigma(p)\},
\eeq
it is tempting to believe that $\si^1 = \sigma(\Theta)$, and that the distribution of $\sigma(\Theta)$ is a mixture of product Bernoulli's with mixing measure given by the distribution of $\Theta$.
To turn this into a rigorous proof is a bit delicate. What follows is one possible approach.

\paragraph*{Step 1:}   
Let $\caP \subset (0,1)$ be a finite set.   
Let $\nu_0=\nu_0^\caP$ be a measure on $\{\pm 1\}^\Z \times (\prod_{p \in \caP} \{\pm 1\}^\Z)$ whose first coordinate is distributed like $\mu$ (this is property (b)) and is independent of the others, while its other $|\caP|$ coordinates are distributed like $\{\Bp:p\in\caP\}$ (property (a)) and are coupled to satisfy property (d). 
\noindent
Also, let $\nu_t$ denote the distribution on $\{\pm 1\}^\Z \times (\prod_{p \in \caP} \{\pm 1\}^\Z)$ obtained by evolving the coupling to time $t$ when started from the initial distribution $\nu_0$.

We now examine a weak$^*$-limit point $\nu_\infty=\nu_\infty^\caP$ of the collection of time averaged measures  $(1/T) \int_0^T \nu_t\,\d t$. The stationarity of $\Bp$ ensures $\nu_\infty$ satisfies (a), the stationarity of $\mu$ ensures (b), and of course (d) is also preserved by attractivity of the coupling. By  \Cref{lem:coupling exists}, $\nu_\infty$ is a stationary coupling, so we may apply  \Cref{L:Dominant} and get property (c). We note that $\si(p)$ is Bernoulli so $\E[r_x(\si(p))^k]<\infty$ for all $k$ and $x$, and ditto for $l_x$. By H\"older's inequality we may thus get from $\E\big[l(\si^1)^{1+\eps}\big]<\infty$ (the assumption of the theorem) that $\E\big[(\maxl_y\minr_y)^{1+\eps}]<\infty$ (the requirement of \Cref{L:Dominant}).

To pass from a finite set to an infinite set we only need to make sure that, when $\caP\subset\caQ$ then $\nu_\infty^\caQ$ is a limit of a subsequence of the sequence that was used to define $\nu_\infty^\caP$. This ensures consistency and allows to apply the Kolmogorov extension theorem. We may thus get a measure coupling $\si^1$ to $\si(p)$ for a dense set of $p$ in $(0,1)$, say to all of the rationals, satisfying (a)-(d). Abusing notations, below when we write $p\in\Q$ we implicitly mean $p\in\Q\cap(0,1)$.

\paragraph*{Step 2:}
Define the random variable
\begin{equation*} \label{def: mystery}
\Theta:= \sup\{p \in \Q : \sigma^1 \geq \sigma(p)\} =   \inf\{p : \sigma^1 \leq \sigma(p)\}
\end{equation*}
The equality follows from monotonicity, and also gives that $\si^1$ has a density i.e.\ the limit
\[
\lim_{L\to\infty}\frac{1}{2L+1}|\{x\in[-L,L]:\sigma^1(x)=1\}|=:\Dens(\sigma^1)
\]
exists and is equal to $\Theta$ since with probability one for all $p\in\Q$ the density of $\si(p)$ is $p$. 
Note that the existence of a density is a property of $\si^1$ irrespective of any coupling of it with anything else.

\paragraph{Step 3:}We now wish to claim that $\Dens(\si^1)$ is independent of $\{\si(p):p\in \Q\}$. For this we need to go through step 1 again. 
Step 1 starts with the measure $\nu_0$ under which $\si^1$ is independent of $\{\si(p):p\in\caP\}$. In particular $\Dens(\si^1)$ is independent of the latter collection of random variables. By  \Cref{lem:density preserved} the density is preserved during the time evolution.  In particular, we may conclude from this that under the measures $(1/T)\int_0^T\nu_t\,\d t$ the density of the first process (``the $\si^1$'') exists and is independent of the other processes (``the $\si(p)$''). 

We need to be careful when passing to the limit in $T$ as, in general, weak$^*$-limits do not preserve the density or its existence. In this case, however, it is only the coupling that changes as we take $T$ to $\infty$, while the marginal distribution of each coordinate stays fixed. To use this fact we need a notion of $\eps$-independence in metric spaces, so let us define it now:
\begin{definition}For two variables $X$ and $Y$ taking values in metric spaces, we say that $X$ and $Y$ are $\eps$-indenedent if for any $A$ and $B$ measurable,
\begin{align*}
\Pr(X\in A)\Pr(Y\in B) &< \Pr(X\in A+\eps, Y\in B+\eps)+\eps\\
\Pr(X\in A,Y\in B) &< \Pr(X\in A+\eps)\Pr(Y\in B+\eps)+\eps
\end{align*}
where $A+\eps$ is the $\eps$-inflation of $A$, namely $A+\eps:=\{x:\exists y\in A \textrm{ such that }d(x,y)<\eps\}$, and ditto for $B+\eps$.
\end{definition}
To use this, find some \emph{local} variable $E$ approximating $\Dens(\si^1)$ up to $\eps$ i.e.\ $\Pr(|E-\Dens|>\eps)<\eps$. We get that $E$ is $\eps$-independent of $\{\si(p):p\in\caP\}$ under $(1/T)\int_0^T\nu_t\,\d t$ for any $T$. Taking a limit we get that $E$ is $2\eps$-independent of $\{\si(p):p\in \caP\}$ under any weak$^*$-subsequential limit (here it was important to have this particular definition of $\eps$-independence). Thus under $\nu_\infty^\caP$, $\Dens(\sigma^1)$ is $3\eps$-independent of $\{\si(p):p\in \caP\}$.  Since $\eps>0$ was arbitrary, we get that $\Theta=\Dens(\sigma^1)$ is independent of $\{\si(p):p\in\caP\}$. This implies directly the property with $\caP$ replaced by $\Q$.



\paragraph{Step 4:}
We now show that $\si^1=\sup\{\si(p):p<\Theta\}$ under $\nu_{\infty}$. To see this, observe that since $\Theta$ is independent of $\{\si(p):p\in[0,1]\}$, we can fix it in advance. On the other hand, for any monotone coupling of Bernoulli processes and any fixed $x$, the random variable $X:=\sup\{ p: \si(p)(x)=-1\}$ is uniformly distributed.  Hence for $\Theta$ fixed, the probability that $\sup\{\si(p):p<\Theta\}\ne\inf\{\si(p):p>\Theta\}$ is zero since none of the individual coordinates of the monotone coupling can flip at $\Theta$.  Further, for any value of $\Theta$ fixed in advance, $\bar{\si}:=\sup\{\si(p):p<\Theta\}$ is a Bernoulli-$\Theta$ process. Since $\si^1$ is sandwiched between $\sup\{\si(p):p<\Theta\}$ and $\inf\{\si(p):p>\Theta\}$, the foregoing discussion implies  that $\si^1=\si(\Theta)$ and $\Theta$ is independent of $\{\si(p)\}$. In other words, $\si^1$ is distributed as a mixture of Bernoulli processes. This proves the theorem. 
\end{proof}
\begin{remark}Surprisingly, perhaps, the standard coupling of Bernoulli processes (i.e.\ the coupling where each site is coupled independently of the others) is not stationary to the Toom process (For example, take a finite system with periodic boundary condition. Then it is straightforward to write a formula for the number of incoming arrows for any given configuration and see that it is not constant). In other words, the process $\si(p)$ which we analyzed during the proof is a collection of Bernoulli processes, coupled by a non-standard monotone coupling. 
\end{remark}

\section{Proof of  Theorem \ref{T:WeakN}}\label{S:WeakN}
\label{S:Bulk}
Recall that we wish to study the process on $\mathbb{N}$ away from 0. For this we consider the coupling process defined on the configuration space $\{\pm 1\}^\N \times \{\pm 1\}^{\Z}$.  We will write $\si=(\si^1(x),\si^2(x))$ for the configuration $ \si \in \{\pm 1\}^\N \times \{\pm 1\}^{\Z}$. We use the notation $\D$, $\D^+$ and $\D^-$ for discrepancies set out at the beginning of \Cref{S:InvZ} with the understanding that \textit{all} vertices $x\leq 0$ host discrepancies at all times.

Throughout this section,  we let $p$ be the unique solution in $(0,1)$ to the equation 
\beq
\label{E:LAP}
\left(\frac{1-p}{p}\right)^2=\frac{\lambda_+}{\lambda_-}\quad.
\eeq

Let $\nu$ be a  probability measure on $\{\pm 1\}^\N\times \{\pm 1\}^{\Z}$ stationary for the coupling process and  with respective marginals $\mu_\infty, \Bp$. We know such measures exist by  \Cref{lem:coupling exists}, see proof in \Cref{sec: invar}. 
(There is in fact a unique such measure, though we do not need to use this explicitly). From the next proposition, \Cref{T:WeakN} follows easily.
\begin{proposition}
\label{L:D-Decay}
With $\nu$ as above, 
\beqs
\lim_{x \rightarrow \infty} \nu(x \in  \mathbf{D} ) =0.
\eeqs
\end{proposition}
Before discussing \Cref{L:D-Decay} further, let us use it to attend to \Cref{T:WeakN}.
\begin{proof}[Proof of \Cref{T:WeakN}]
Consider a local function $f$ and recall the space shifts $\tau_x$ (defined on page \pageref{T:WeakN}). By the definition of $\mathbf{D}$; 
\[
\left|\E_{\mu_{\infty}}[f \circ \tau_x]- \E_{\Bp}[ f\circ \tau_x]\right|\leq \| f\|_{\infty}\sum_{y \in x+ \supp f } \nu(y \in  \mathbf{D} ).
\]
By \Cref{L:D-Decay} the RHS tends to $0$ as $x$ tends to infinity.  But since $\Bp$ is invariant under spatial shifts, this implies the push forward of $\mu_{\infty}$ by $\tau_x$ converges weakly to $\Bp$ as $x$ tends to infinity.
\end{proof}

\begin{proof}[Proof of \Cref{L:D-Decay}]
As in the proof of  \Cref{L:Discrep1}, the crucial quantity that we will analyze is the \emph{discrepancy flow through a point $x$}, which we will again denote by $j_x$. 
Recall that $j_x$ is the averaged infinitesimal rate at which a discrepancy in $(-\infty,x]$ jumps to $(x,\infty)$. Let us note a few simple properties of $j$. First,
$
\nu(x\in\mathbf{D})\le j_{x}
$
because whenever there is a discrepancy at $x$ it jumps to $(x,\infty)$
with rate at least the rate of the Poisson arrivals at $x$. 
Next, note that $j_x$ is always finite: both $\sigma^1$ and $\sigma^2$ satisfy the local \conda, $\sigma^1$ because $l_x(\si^1)\le x$ and in particular is finite, while $\si^2$ is a Bernoulli process. 
By  \Cref{lem:ja finite}, this implies that $j_x$ is finite.

Next, as  in the proof of  \Cref{L:Discrep1},  let $a_x$ be the rate of annihilations at $x$, i.e.\ the averaged infinitesimal rate at which a $+$ discrepancy at some $y<x$ moves on a $-$ discrepancy at $x$ causing both to disappear, or vice versa.  Also let us recall \cref{eq:ja} i.e.\ 
$j_{x-1}-j_x=a_x$.
Let $k$ be some fixed parameter and let $x\in\mathbb{N}$. One of the following
three events must occur:
\begin{enumerate}
\item The interval $I_{x}:=[\max(0,x-k^{2}),x]$ contains two discrepancies
with opposing signs.
\item The last $k$ discrepancies before $x$ are of the same sign (and
case 1 did not happen).
\item The last $k$ discrepancies before $x$ are not of the same sign (and
case 1 did not happen)
\end{enumerate}
We divide the set of discrepancies $\mathbf{D}$
and the infinitesimal rates
of discrepancy flow $j_{x}$  into three parts accordingly, $\mathbf{D}=\mathbf{D}^{1}\cup\mathbf{D}^{2}\cup\mathbf{D}^{2}$ 
and $j_{x}=j_{x}^{1}+j_{x}^{2}+j_{x}^{3}$. Thus, for example,\label{page:j123}
\[
\mathbf{D}^{1}=\{x\in\mathbf{D}:I_{x}\cap\mathbf{D}^{+}\ne\emptyset,\; I_{x}\cap\mathbf{D}^{-}\ne\emptyset\}
\]
and similarly for the other $\mathbf{D}^{i}$. Let us also give an example
with $j$: $j_{x}^{1}$ is the averaged infinitesimal rate that discrepancies jump from $(-\infty,x]$
to $(x,\infty)$ while there are two discrepancies of opposite signs
in $I_{x}$ (we are not bothered about whether the discrepancy that
performed the jump is in fact one of those).

Let us isolate the estimates
used to control $\mathbf D$ and $j$ in two lemmas.   We note first of all that 
\[
\nu(x\in\mathbf{D}^{i})\le j_{x}^{i}\qquad i=1,2,3.
\]
following the same argument we used above to show that $\nu(x\in\mathbf{D})\le j_x$.  This means that our control on $\mathbf{D}^i$ in these lemmas follows from our control on $j^i$.

\begin{lemma}
\label{lem:ED1}For every $k$ there exists a $C_{k}$ such that $\mathbb{E}|\mathbf{D}^{1}|\le C_{k}$.
Furthermore, 
\[
\sum_{x=1}^{\infty}j_{x}^{1}\le C_{k}.
\]

\end{lemma}
The proof (see below) proceeds by observing that a configuration contributing to $j^1_x$ has
a positive probability to lead to an annihilation. This part of the argument would work equally well with any value of $p$. It is in the next lemma that the specific choice \Cref{E:LAP} becomes crucial.

\begin{lemma}
\label{lem:ED2}There exist numbers $\epsilon_{k}\downarrow 0$ such that
$\nu(|\mathbf{D}^{2}\cap[1,n]|)\le\epsilon_{k}n$. Further, $j_{x}^{2}\le\epsilon_{k}$
for all $x> k$.
\end{lemma}

We defer the proof of both lemmas and instead first show how they imply \Cref{L:D-Decay}. We start with a preliminary claim
\begin{lemma}
\label{lem:averagedD}$\lim_{n\to\infty}\frac{1}{n}\nu(|\mathbf{D}\cap[1,n]|)=0$
\end{lemma}
\begin{proof}
Fix $k\in \N$ and decompose $\D$ into the $3$ parts using this $k$.  By the ``$\mathbf{D}$ parts'' of
\Cref{lem:ED1,,lem:ED2} we see that 
\[
\nu\big(\big|(\mathbf{D}^{1}\cup\mathbf{D}^{2})\cap[1,n]\big|\big)
\le C_{k}+n\epsilon_{k}.
\]
As for $\mathbf{D}^{3}$, the conditions of case 3 imply that the
interval $[\max(0,x-k^{2}),x]$ contains less than $k$ discrepancies.
It follows from this that $|\mathbf{D}^{3}\cap[1,n]|\le n/k+k^2$.  This implies
\[
\limsup_n (1/n) \nu(|\D\cap [1, n]|)\leq \eps_k+1/k.
\] Taking
$k\to\infty$, the lemma is proved.
\end{proof}
To move from the averaged result of  \Cref{lem:averagedD} to
the non-averaged result stated in \Cref{L:D-Decay}, we argue as follows. Let $\epsilon>0$. First
fix some $k$ sufficiently large such that the $\epsilon_{k}$ from
 \Cref{lem:ED2} satisfies $\epsilon_{k}<\epsilon$. By 
 \Cref{lem:ED1} we can choose $x_{0}$ such that for all $x>x_{0}$
one has that $j_{x}^{1}<\epsilon$. Next, let $M$ be some parameter (to
be fixed shortly and depending only on $k$ and on $\lambda^{+}/\lambda^{-}$).
By  \Cref{lem:averagedD} we can further find some $x>x_{0}$ such that
\begin{equation}
\sum_{y=x-M}^{x}\nu(y\in\mathbf{D})<\epsilon/k\label{eq:def x with M}
\end{equation}
(we assume here that $x>M$ and, while we are at it, also $x>k^{2}$). 

We will
now show that $j_{x}<C\epsilon$.
Since $j_{x}^{1}$ and $j_{x}^{2}$ are already given to us as satisfying such a bound we need
only estimate $j_{x}^{3}$. For this purpose we note that $\sigma^{2}$
is a Bernoulli process, so there are constants $c, C, C'>0$, such that one cannot create an interval
of consecutive equal signs larger than $Ck$ before $x$ by
changing the sample $\si^2$ at less than $k$  locations, except  with probability bounded by $C'e^{-ck}$.  Fix the
parameter \emph{$M$} from the last paragraph to be this $Ck$ and
denote this event by $G_{k}$. Split $j_{x}^{3}=j_{x}^{3.1}+j_{x}^{3.2}$, where
 $j_{x}^{3.1}$
denotes the rate at which discrepancies pass from $(-\infty,x]$ to $(x,\infty)$
by arrivals in $[x-M,x]$ and $j_{x}^{3.2}$ is the remainder. It
follows that
\begin{equation}
j_{x}^{3.1}\le(M+1)\sum_{y=x-M}^{x}\nu(y\in\mathbf{D}).\label{eq:jx31}
\end{equation}
because contributions to $j_{x}^{3.1}$ require at least one discrepancy in $[x-M,x]$ and also require a Poisson arrival in this interval.  The factor $M+1$ in the estimate comes as an upper bound on the latter rate. Using \Cref{eq:def x with M}
and the fact that $M=Ck$ we get $j_{x}^{3.1}\le C\epsilon$.

To bound $j_{x}^{3.2}$ notice that a spatial interval such that both configurations have a pair of particles of opposing
signs acts as a barrier for discrepancies.
Hence if an arrival ocurred at some $y<x$ and pushed a discrepancy
beyond $x$, either $\sigma_{1}$ or $\sigma_{2}$ must have
constant sign between $x$ and $y$. As we are in case $3$ and two discrepancies
of opposing signs constitute a barrier to discrepancy motion, the discrepancies between 
between $x$ and $y$ must be of the same type and there must be less than $k$ of them. We conclude that $\sigma_{2}$
must have less than $k$ particles of some kind. This allows to bound
$j_{x}^{3.2}$ by examining $\sigma_{2}$ only, and $\sigma_{2}$
is a Bernoulli process. 

The probability that $\sigma_{2}$ restricted
to $[y,x]$ has less than $k$ particles of some kind can be bounded
above roughly by $C\exp(-ck-c(x-y)/k)$ --- recall the definition
of $M$.  We get, by definition of $j_{x}^{3.2}$, that
\[
j_{x}^{3.2}\le C\sum_{y=-\infty}^{x-M-1}e^{-ck-c(x-y)/k}\le Ce^{-ck}.
\]
Taking everything together we get
\[
j_{x}\le C\epsilon+Ce^{-ck}.
\]
and since we assumed $k$ is sufficiently large (depending on $\epsilon$)
we may incorporate the $Ce^{-ck}$ into the other term and conclude,
as promised, that $j_{x}\le C\epsilon$.

\Cref{L:D-Decay} is now proved because $j_{x}$ are decreasing, by \Cref{eq:ja}. We get
that for all $y>x$ $j_{y}\le C\epsilon$ and since $\epsilon$ was
arbitrary $j_{y}\to0$. As already remarked $\nu(x\in\mathbf{D})\le j_{x}$
so we also get $\nu(x\in\mathbf{D})\to0$.
\end{proof}

\begin{proof}[Proof of  \Cref{lem:ED1}]
An important corollary of \Cref{eq:ja} is that $j_x$ is decreasing and that $\sum a_x<\infty$. The proof of the lemma then follows by comparing $j_x^1$ and $a_x$. $j_x^1$ is the rate at which discrepancies flow through $x$ while there are two discrepancies with opposing signs in the interval $I_x=[\max(0,x-k^2),x]$. This couple of discrepancies acts as a barrier, so any arrival that pushed a discrepancy beyond $x$ must be in $I_x$. Hence
\begin{equation}\label{eq:jx1}
j_x^1\le|I_x|\cdot\nu(\mbox{there are two discrepancies of opposing signs in $I_x$}).
\end{equation}
In the language of   \Cref{eq: in the eye of the hunter}, this event implies that for some $y,z\in I_x$, the event $E_{y,z}$ occurred. Applying  \Cref{eq: in the eye of the hunter} gives
\[
j_x^1
\stackrel{\textrm{\eqref{eq:jx1}}}{\le}
k^2\sum_{y,z\in I_x}\nu(E_{y,z})\le \frac{k^2}{c(k)}\sum_{y,z\in I_x}\sum_{w=y}^z a_w
\le C(k)\sum_{w\in I_x}a_w.
\]
The claim of \Cref{lem:ED1} follows from $\sum a_x<\infty$.
\end{proof}
\subsection{Proof of  Lemma \ref{lem:ED2}}
As already mentioned, it is in this part of the proof that the relation between $p$ and $\lambda$ is used. We start by indicating the reason for this relation.  First denote by $H_x^\pm(t)$ the total number of $\pm$ discrepancies which jumped from $(-\infty,x]$ to $(x,\infty)$ between time 0 and time $t$. Let $H_x(t)=H_x^+(t)+H_x^-(t)$ be the total flow and let $K_x(t)=H_x^+(t)-H_x^-(t)$ be the signed flow through $x$. We will only need the signed flow at 0 (where it is a function of the Bernoulli process alone) and so abbreviate $K(t)=K_0(t)$.
\begin{lemma}If $((1-p)/p)^2=\lambda_+/\lambda_-$ then $\E[K(t)]=0$.
\end{lemma}
\begin{proof}This is a straightforward calculation: the infinitesimal rate at which $+$ particle enter $\N$ is $\la_+l_1(\si^2)\mathbbm{1}\{\si^2(0)=1\}$ (recall that $l_1$ is the length of the block of identical spins from $0$ to its left). Since $\si^2$ is a Bernoulli-$p$ process, $\E[l_1(\si^2)\cdot\mathbbm{1}\{(\si^2(0)=1\}]=p/(1-p)$. The same holds for the $-$ particles and we get $\E(K(t))=\frac{p\la_+}{1-p}-\frac{(1-p)\la_-}{p}$.
\end{proof}
We next need two preliminary results which are interesting in their own right.
It is possible to show that $K$ satisfies a functional CLT under proper rescaling, see \cite{CKR2}, but for our current purposes, the following diffusive bound suffices. Its proof is supplied in \Cref{S:proofspi}.
\begin{lemma}
\label{L:FDD1}
There exists $C>0$ such that for all $t \in \R$,
\[
\E_{\Bp}[(K(t)-K(0))^2] \leq C t.
\]
\end{lemma}

 We also need a bound on the total flux of discrepancies across a fixed vertex for the process on $\{\pm 1\}^\Z$.
Set $\NN([t,s]) =H_0(t)-H_0(s)$. 
\begin{lemma}
\label{L:ExpMom}
For  sufficiently small  $\gamma$,
\[
\E_{\Bp}[\exp(\gamma \caN(I)/\str I \str)] \leq  C, \qquad \text{for any $I$}.
\]
In particular 
\beq \label{eq: crude current}
\bbP_{\Bp} (\NN(I)  \geq N )  \leq C(m)  (\str I \str/{N})^{m}
\eeq
\end{lemma}
The proof is also deferred to \Cref{S:proofspi}.

To explain the relation between $j_x^2$ and $K$ we need the following definition. For a partition $\pi$ of $[0, t]$ into intervals, let 
\[
\KK(\pi)=\sum_{I \in \pi}\str \Delta K (I)  \str \mathbbm 1\{ \str \Delta K (I)  \str \geq k\}
\]
where  $\Delta K(I):= K(\sup I)-K(\inf I)$ and let
\beqs
\KK^*(t)  = \KK^*(k,t)=\sup_{\pi}  \KK(\pi).
\eeqs
With these definitions we can now claim
\begin{lemma}For all $x\geq 1$ and $t$, $H_x^2(t)\le x+\KK^*(t)$.
\end{lemma}
Here $H_x^2(t)$ is the number of discrepancies of ``type 2'' which passed through $x$ until time $t$, using the same classification of discrepancies into 3 types we used on page \pageref{page:j123} to define $j_x^2$.
\begin{proof}
The $x$ term in the lemma is a crude bound for ``original'' discrepancies i.e.\ for discrepancies which existed at time 0 in $[1, x]$. We ignore these discrepancies and order the others by their time of crossing $x$. 

In this proof it will be convenient to think about discrepancies as having a fixed order. Examine a discrepancy. When it moves (to the right only!), it can collide with the first discrepancy to its right.  If the second discrepancy has the opposite sign, the two annihilate leaving behind two vertices each having spin agreement.  If the two discrepancies have the same sign we shall use the interpretation that the first one takes the place of the second one, and the second starts moving. This can create a chain reaction, but we note that as soon as there is an annihilation, the process stops.  In any case, \textit{the order of discrepancies never changes}.   

Recall the notation $I_x=[\min(0, x-k^2), x]$.  Let us now use this interpretation to examine a stretch (in time) of $\ell$ discrepancies in $\mathbf{D}_2$ contributing to $H^2_x(t)$, say with sign $+$ for concreteness. The clock ring precipitating a jump contributing to $H^2_x(t)$ does not have to occur in $I_x$ and could come from a particle of either sign.   However, the chain reaction which occurs can only involve the discrepancies with sign $+$.  This follows from our definition; just before the jump occurred, the last $k$ discrepancies to the left of $x$ were of the same sign. 
Since discrepancies of different types do not propagate motion when they collide, we conclude that the discrepancies of $I_x$ must be $+$ discrepancies.

Hence a stretch of $\ell$ discrepancies of type 2 crossing $x$  corresponds to $\ell+k-1$ discrepancies of sign $+$, the $\ell$ discrepancies which crossed as well as another (at least) $k-1$ which remain at the time of the last crossing. Let $a$ be the time when the first one entered the system (from $(-\infty, 0]$) and $b$ the time when the last one did.  Such a labeling makes sense because of our interpretation of the discrepancy dynamics and since discrepancies only move to the right in $[1, \infty)$.

We claim that $K(b)-K(a)= \ell +k-1$. This relies on the fact that signed discrepancy sum is preserved by our dynamics: discrepancies survive until they are annihilated in pairs of $+$ and $-$. So if we see $\ell+k-1$ consecutive $+$ discrepancies at some space-time point $(x, t)$, we must have started with a signed sum of $\ell+k-1$. (the term ``consecutive'' might be slightly misleading here, since some of them are consecutive in the time of crossing of $x$ and others are consecutive in space at a given time, but for the claim on $K(b)-K(a)$ this does not matter).

Finally, different stretches of discrepancies of type 2 crossing $x$ must correspond to different time intervals $[a,b]$. Indeed, examine the last $+$ discrepancy in one stretch and the first $-$ discrepancy in the next. Being of different sign they cannot cross one another without annihilating, and since we know the $+$ arrived at $I_x$ before the $-$, the $-$ must have started after it.

Thus the stretches of discrepancies of type 2 crossing $x$ form disjoint collection of subintervals of $[0,t]$, each of which has $|\Delta K(I)|\ge k$. Completing this collection to a partition $\pi$, the lemma is proved.
%
%
\end{proof}

 \Cref{lem:ED2} will thus be proved once we show
\begin{lemma}
\label{L:LargeFluc}
\beqs
\limsup_{k\to\infty}  \limsup_{t\to\infty} \frac{1}{t}\E_{\Bp}[ \KK^*(k,t)]  = 0.
\eeqs
\end{lemma}
As before, we have written $\E_{\Bp}$ instead of $\E$ since $K$ depends only on $\si^{2}$.
\begin{proof}
The proof of this lemma relies on a separation of scales.  
For each partition $\pi$ of $[0, t]$, let us split the sum $\KK(\pi)$ according to interval sizes:
\begin{align}
\KK(\pi) &=\underbrace{\sum_{I\in \pi: |I|\leq k^{1/2}}   |\Delta K(I) | \mathbbm{1}\{ \str \Delta K(I)  \str \geq k\}}_{\KK^1(\pi)} + 
\underbrace{\sum_{I \in \pi: |I|  > k^{1/2}} |\Delta K(I) | \mathbbm{1}\{ |\Delta K(I)   \str \geq k\} }_{\KK^2(\pi)}
\end{align}
We will use separate mechanisms to bound each of $\KK^1$ and $\KK^2$ uniformly in $\pi$.

Let us first attend to $\KK^1(\pi)$.  The idea here is simple:  since a contributing interval $I$ is small relative to $k$, it necessitates too many Poisson arrivals, at least $k$, in $I$.  This is a rare event, and gets exponentially rarer as $|I| \rightarrow 0$.  This fact allows us to handle all partitions simultaneously via a properly chosen infinite covering of $[0, t]$.

For each $j \in \mathbb Z$ let 
\[
X_j=  \sum_{i=0}^{\lfloor t2^{-j}\rfloor} \NN([i2^j,(i+1)2^j])\mathbbm 1\{ \NN([i2^j,(i+1)2^j]) \geq k/2 \}.
\]
Then for any $\pi$, exploiting that $\NN(I)$ is the total variation of $K(I)$,
\[
\sup_{\pi}  \KK^1(\pi) \leq 2\sum_{j \leq \log_2( k^{1/2})+1}  X_j.
\]
Taking expected values on both sides and using stationarity
\begin{equation}
\label{E:Small}
\frac 1t \E_{\Bp}\left[\sup_{\pi} \KK^1(\pi)\right] \leq \sum_{j \leq 
\log_2(k^{1/2})+1} 2^{1-j} \E_{\Bp}\left[\NN([0, 2^j])\mathbbm{1}\{ \NN([0, 2^j]) \geq k\}\right].
\end{equation}
Using \Cref{eq: crude current}, we find that for $k$ sufficiently large, the RHS of \Cref{E:Small} is summable and that, moreover, it tends to $0$ as $k$ tends to $\infty$. 

To bound $\KK^2(\pi)$ let us introduce a reference partition $\rho=\{[ik^{1/4},(i+1)k^{1/4}]\}$ of $[0, t]$ (shorten the last interval if necessary). Let $\rho'\subset\rho$ be the collection of intervals which contain an endpoint of some interval from $\pi$. Then we can write 
\[
\KK^2(\pi)\le
\underbrace{\sum_{I\in \rho}|\Delta K(I)|}_{\mathrm I}
+2\underbrace{\sum_{I\in\rho}\caN(I)\mathbbm{1}\{\caN(I)\ge k^{1/3}\}}_{\mathrm{II}}
+2\underbrace{\sum_{I\in\rho'}\caN(I)\mathbbm{1}\{\caN(I)< k^{1/3}\}}_{\mathrm{III}}.
\]
Term $\mathrm{ III}$ is bounded by $ 4t k^{-1/2} k^{1/3}$ since the total number of intervals $I\in \pi$ with $|I| \geq k^{1/2}$ is at most $2 tk^{-1/2}$. 
For the other two terms, we take the expectation and use the fact that they no  longer depend on $\pi$. 
For $\mathrm{ II}$, the bound \Cref{eq: crude current}, yields
$$
\E_{\Bp} [\mathrm{ II}]   \leq   2 tk^{-1/4} \times   C(m) (k^{1/4}/k^{1/3})^m    
$$
which vanishes as $k\to \infty$ by choosing $m$ sufficiently large.  For the first term, we argue
\begin{align}
\label{E:K}  \E_{\Bp} [\mathrm{I}] \leq t k^{-1/4}   (\E_{\Bp}[|K(k^{1/4})-K(0)|^2])^{1/2} \leq C t  k^{-1/8}
\end{align}
where the  first inequality is  by stationarity and Cauchy-Schwarz, and the second follows from the diffusive moment estimate of  \Cref{L:FDD1}.
Hence we have obtained $$\limsup_{k \to \infty}  (1/t) \E_\Bp\sup_{\pi} \KK^2(\pi) =0,   $$
Combining with the analogous bound on $\KK^1 $, the assertion of the lemma follows.
\end{proof}

\section{Auxiliary results}\label{S:auxiliary}
\subsection{Proofs of  lemmas \ref{L:FDD1} and \ref{L:ExpMom}.}
\label{S:proofspi}

\begin{proof}[Proof of  \Cref{L:FDD1}]
Let $f$ be the infinitesimal drift of $K$ i.e.\ $f(t)=\si(0)\lambda_{\si(0)}l_1(\si)$ (in other words, if $\si(0)=1$ then $f=\lambda_+l_1$ and otherwise $f=-\lambda_-l_1$). Using the Cauchy-Schwarz inequality,
\[
\E\big[K(t)^2\big]\le 
2\E\Big[\Big(K-\int_0^tf(\si_s)\,\d s\Big)^2\Big]+
2\E\Big[\Big(\int_0^tf(\si_s)\,\d s\Big)^2\Big].
\]
(recall that $K$ and hence $f$ are functions of the Bernoulli process, so all expectations are likewise with respect to it). The first term is a martingale, so its expected square is its quadratic variation.   By direct calculation this is $\E\left[\int_0^t\lambda_{\si(0)}l_1(\si)\right]$, which is bounded by $Ct$. For the second term, we use stationarity to bound it by
\[
2t \int_0^{\infty}  \str \E[f(\sigma_s)f(\sigma_0)] \str\,\d s.
\]
We shall show $\E_{\Bp}[f(\sigma_s) f(\sigma_0 )]$ decays exponentially in $s$ to complete the proof.  

We want to apply \Cref{thm: correlation decay}.  One slight complication is that the function $f$ has unbounded support.  To handle this issue,  let $f_n$ be the approximation to $f$ above by replacing $l_1$ with $\min\{l_1,n\}$, which is a local function. 
Then, by inspection, 
\[
\E[(f-f_n)^2] \leq C e^{-cn} 
\]
so that, for any $n$, Cauchy-Schwarz yields
\[
\left|\E[f(\sigma_s)f(\sigma_0)]-\E[f_n(\sigma_s)f_n(\sigma_0)]\right| \leq Ce^{- cn/2}.
\]
The autocorrelation of $f_n$ is handled by \Cref{thm: correlation decay}, with $r=n$, so that 
\[
\left|\E[f(\sigma_t)f(\sigma_0)] \right| \leq C (e^{-cn/2} +   \e^{n - ct})
\]
and, choosing $n=\frac12 ct$, we get 
\[
\left|\E[f(\sigma_t)f(\sigma_0)]\right|\leq C\e^{-ct}.
\]
  \Cref{L:FDD1} follows.
\end{proof}

\begin{proof}[Proof of  \Cref{L:ExpMom}]
Recall that we wish to estimate $H=H_0(t)$, the number of particles entering $\N$ in $[0,t]$. We define $H^L(t)$ to be the number of particles which enter from $(-L,0)$, and note that $H^L\nearrow H$ so by monotone convergence $\E(\exp(\gamma H^L))\to\E(\exp(\gamma H))$.

Next denote the drift of $H^L$ by $v=v^L(t)$ i.e.
$$
v^L_t:=
\int_0^t  \min(l_1(\si_s),L)\,\d s. 
$$
We first find an estimate on $\E(\exp(\alpha v))$.
We Taylor expand the exponential 
\[
\E(\e^{ \alpha v} )=\sum_{k=0}^\infty \frac{\alpha^k}{k!}\E(v^k),
\]
integrate each term and use H\"older's inequality and stationarity to get 
\[
\E(v^k)\le\int_0^t\dotsb\int_0^t\E(l_1(\si_{s_1})\dotsb l_1(\si_{s_k}))\,\d s_1\dotsb\d s_k \le t^k\E(l_1(\si_{0})^k)
\]
so 
\begin{equation} \label{eq: bound predict}
\E(\e^{ \alpha v} )
\leq   \E_\Bp[\e^{\alpha t l_1}].
\end{equation}
which is finite for  $\alpha< c(p)/t$. 

Next, one verifies by direct computation (omitted here) that, for any $\beta>0$,
\[
Z^L_{\beta}(t) :=\exp\left\{ \beta H^L(t)-  (e^{\beta}-1)v^L_t \right\}
\]
is a martingale.   To exploit this, fix $\beta\in \R$ and set $\al=(1/2)(e^{2\beta}-1)$ and write
\[
\E_{\Bp}[\exp(\beta H^L)]=\E_{\Bp}\left[\exp(\beta H^L- \al v) \exp(\al v)\right].
\]
Applying the Cauchy-Schwarz inequality and noting that $2\al=e^{2\beta}-1$
\[
\E_{\Bp}[\exp(\beta H^L(t))] \leq \E_{\Bp}[Z_{2\beta}(t)]^{1/2} \E_{\Bp}[\exp(2\al v)]^{ 1/2}.
\]
Since the first term on the RHS evaluates to $1$ (being the expectation of a martingale), the lemma follows from \cref{eq: bound predict} and $\E(\exp(\al H^L))\to\E(\exp(\al H))$ 
upon choosing $\beta= \ga/t$ and noting that if $\ga$ is sufficiently small then $2\al t$ is small enough to ensure that the right hand side of \Cref{eq: bound predict} is finite.
\end{proof}


\subsection{Existence of a stationary coupling}\label{sec: invar}
Our goal in this section is to prove  \Cref{lem:coupling exists}. Recall that it stated that if $(\mu^1,F^1)$ and $(\mu^2,F^2)$ are two stationary, regular Toom processes and if $\nu$ is any coupling of them, then 
any subsequential weak$^*$-limit of $\frac 1T\int_0^T\nu_t$ is a stationary coupling of $(\mu^1,F^1)$ and $(\mu^2,F^2)$. Here $\nu_t$ is the result of applying the coupling to $\nu$ for time $t$. As this is the only place in the paper where regularity of Toom processes is used, let us recall its definition (\Cref{def:regular} on page \pageref{def:regular}): a Toom process $(\mu,F)$ is called regular if $F_t$ can be written as the $\mu\times\Pr$-limit in measure of functions $F_t^L$ such that $F_t^L(\cdot,\omega)$ is continuous for almost all $\omega$.
\begin{proof}As the lemma is used both to construct a coupling between two processes on $\Z$ (in \Cref{S:InvZ}) and between a process on $\Z$ and a process on $\N$ (in \Cref{S:WeakN}), let us denote by $\Xi$ the corresponding state space i.e.\ $\big(\{\pm 1\}^\Z\big)^2$ or $\{\pm 1\}^\N\times\{\pm 1\}^\Z$, as the case may be. Let $T_n$ be the sequence over which we assume said measures converge weakly$^*$ and denote
\[
\lambda_n=\frac{1}{T_n}\int_0^{T_n}\nu_t\,\d t\qquad \lambda_\infty=\lim\lambda_n.
\]
Fix some $t$. By the definition of regularity, we may write $F_t^1$ as a limit in measure of functions $F_t^{1,L}$ such that $F^{1,L}_t(\cdot,\omega)$ is continuous for almost all $\omega$. Repeat for $F^2$. Let us define the corresponding operators on probability measures on $\Xi$ i.e.
\begin{align*}
P_t\lambda(f)&=\int f(F_t(\eta^1,\eta^2,\omega))
\,\d\lambda(\eta^1,\eta^2)\,\d\mathbbm P(\omega)&&\forall f:\Xi\to\R \textrm{ continuous}\\
P_{t}^L\lambda(f)&=\int f(F^L_t(\eta^1,\eta^2,\omega))
\,\d\lambda(\eta^1,\eta^2)\,\d\mathbbm P(\omega)&&\forall f:\Xi\to\R \textrm{ continuous.}
\end{align*}
where here and below we denote
\[
F_t(\eta^1,\eta^2,\omega)=(F_t^1(\eta^1,\omega),F_t^2(\eta^2,\omega))
\]
and ditto for $F^L$.
For conciseness, fix $t$ and remove it from the notations $P_t$ and $P_{t}^L$. We now fix some continuous $f:\Xi\to\R$ and write
\begin{align*}
\lefteqn{P\lambda_\infty(f)-\lambda_\infty(f)=}\qquad&\\
&(P\lambda_\infty(f)-P^{L}\lambda_\infty(f))
+(P^L\lambda_\infty(f)-P^L\lambda_n(f))
+(P^L\lambda_n(f)-P\lambda_n(f))\\
&+(P\lambda_n(f)-\lambda_n(f))
+(\lambda_n(f)-\lambda_\infty(f))=\textrm{I}+\textrm{II}+\dotsb+\textrm{V}.
\end{align*}
We will now bound the different terms. We need the following lemma.
\begin{lemma}\label{lem:uniformly in coupling}For every continuous $f:\Xi\to\R$ and every $\eps>0$ there exists an $L$ such that
\begin{equation}\label{eq:FFL}
\int |f(F^L_t(\eta^1,\eta^2,\omega))
-f(F_t(\eta^1,\eta^2,\omega))| \,\d\lambda(\eta^1,\eta^2)\,\d\mathbbm P(\omega)\le\eps
\end{equation}
for any $\lambda$ with marginals $\mu^1$ and $\mu^2$ and for every $t>0$.
\end{lemma}
\begin{proof}Let $\delta>0$ be some parameter to be fixed later. Let $B^1\subset \Xi\times\varOmega$ be the set
\[
B^1=\{(\eta^1,\eta^2,\omega):|F^{1,L}_t(\eta^1,\omega)-F^1_t(\eta^1,\omega)|>\delta\}
\]
and similarly $B^2$ with $F^2$ instead of $F^1$. Since $F^{1,L}\to F^1$ in measure, and since $B^1$ does not depend on $\eta^2$ or on the coupling, we get that for $L$ sufficiently large $\nu\times\mathbbm P(B^1)<\delta$, and similarly for $B^2$.
Examine now the integral in \Cref{eq:FFL} and write
\[
\int=\int_{B^1}+\int_{B^2}+\int_{\Xi\times\varOmega\setminus(B^1\cup B^2)}.
\]
The first and second terms are each bounded by $\delta\|f\|_{\infty}$ since the measures of $B^i$ are small. The last term is bounded by the modulus of continuity of $f$ i.e.\ by $\max\{|f(\eta)-f(\eta')|:d(\eta,\eta')\le 2\delta\}$. If we pick $\delta$ sufficiently small such that the sum of the three terms is smaller than $\eps$, and the lemma is proved.
\end{proof}
We return to bounding the terms $\mathrm{I},\dotsc,\mathrm{V}$. To bound terms I and III we use  \Cref{lem:uniformly in coupling} to get that $\mathrm{I,III}\le \eps$ whenever $L$ is sufficiently large (depending on $f$ and $\eps$, but independent of $n$). Term V converges to 0 as $n\to\infty$ because $\lambda_n\to\lambda_\infty$ weakly$^*$ and $f$ is continuous. For term II we write
\[
\mathrm{II}=
\int \bigg(\int f(F^L_t(\eta^1,\eta^2,\omega))\,\d\lambda_\infty(\eta^1,\eta^2)
- \int f(F^L_t(\eta^1,\eta^2,\omega))\,\d\lambda_n(\eta^1,\eta^2)\bigg)\,\d\mathbbm P(\omega).
\]
The functions integrated are continuous (for almost every $\omega$), so the inner term converges to 0 for almost every $\omega$ as $n\to\infty$. It is also bounded since $f$ is bounded. By the bounded convergence theorem we get that $\mathrm{II}\to 0$ as $n\to\infty$.

Finally, term IV is bounded by the observation that $P_t\nu_s=\nu_{t+s}$. Let us postpone the proof of this fact (which is just playing with definitions) and write
\begin{align*}
\textrm{IV}&=\frac{1}{T_n}\bigg(P_t\int_0^{T_n}\nu_s\,\d s-\int_0^{T_n}\nu_s\,\d s\bigg)(f)
=\frac{1}{T_n}\bigg(\int_0^{T_n}\nu_{s+t}\,\d s-\int_0^{T_n}\nu_s\,\d s\bigg)(f)\\
&=\frac{1}{T_n}\bigg(\int_{T_n}^{T_n+t}\nu_s\,\d s-\int_0^{t}\nu_s\,\d s\bigg)(f)\le \frac{2t}{T_n}\max |f|\xrightarrow[n\to\infty]{} 0.
\end{align*}
Let us wrap up the calculation. We fix $L$ sufficiently large so that $\mathrm{I}+\mathrm{III}\le\eps$ uniformly in $n$. We take $n\to\infty$ and get $|P\lambda_\infty(f)-\lambda_\infty(f)|\le\eps$. Since $\eps$ was arbitrary, these are actually equal. Since $f$ was an arbitrary continuous function, $P_t\lambda_\infty=\lambda_\infty$. Since $t$ was arbitrary, this is the required stationarity.

We still need to show $P_t\nu_s=\nu_{t+s}$. We write
\begin{align*}
P_t\nu_s(f)&=\int f(F_t(\eta^1,\eta^2,\omega))\,\d\nu_s(\eta^1,\eta^2)\,\d\mathbbm P(\omega)\\
&=\int f(F_t(F_s(\eta^1,\eta^2,\omega'),\omega))\,\d\nu(\eta^1,\eta^2)\,\d\mathbbm P(\omega')\,\d\mathbbm P(\omega)\\
&=\int f(F_t(F_s(\eta^1,\eta^2,\omega),S_s\omega))\,\d\nu(\eta^1,\eta^2)\,\d\mathbbm P(\omega)\\
&=\int f(F_{t+s}(\eta^1,\eta^2,\omega))\,\d\nu(\eta^1,\eta^2)\,\d\mathbbm P(\omega)
=\nu_{t+s}(f).
\end{align*}
The first equality is the definition of $P_t$. The second is the definition of $\nu_s$. The third uses that $F_s$ is $\caF_s$-measurable, so we can replace the two independent Poisson processes $\omega$ and $\omega'$ with a single Poisson process, and use $\omega$ for the first and $S_s\omega$ for the second (recall that $S_s$ are the natural time shifts of the Poisson process). The fourth equality is the semigroup property for $F$. 
\end{proof}


\subsection{Proof of Lemma \ref{lem:quenched generator}}\label{sec: diff coupled1}
Recall that the statement of the lemma is that rates are well-defined for local functions and Toom processes satisfying the local \conda.
We begin this section with a technical lemma.   Let $Q(y,t)$ be the number of times $\sigma_{s}(y)$
changed in the time interval $[0,t]$ and $N_y$ the Poisson process at $y$.
\begin{lemma}
\label{lem:Ql}Let $\sigma$ be a Toom process, $y\in\mathbb{Z}$ and $t>0$. Then
\[
\mathbb{E}Q(y,t)\le\int_{0}^{t}\mathbb{E}l_{y}(\sigma_{s})\, ds  + \E N_y(t)
\]
(possibly in the sense of $\infty\le\infty$).\end{lemma}
\begin{proof}
Let $x<y$ and let $Q(x,y,t)$ be the number of Poisson arrivals in
$x$ in the time interval $[0,t]$ that caused the value of $\sigma(y)$
to change (recall from the definition of a Toom process that every
change in $y$ must correspond to some $x<y$ and some Poisson arrival
at $x$ at time $t$ such that $\sigma_{t^{-}}(x)=\sigma_{t^{-}}(x+1)=\dotsb=\sigma_{t^{-}}(y-1)=-\sigma_{t^{-}}(y)$).
Then
\[
\mathbb{E}Q(y,t)=\sum_{x<y}\mathbb{E}Q(x,y,t)+ \E N_y(t)
\]
and 
\[
\int_{0}^{t}\mathbb{E}l_{y}(\sigma_{s})\, ds=\sum_{x}\int_{0}^{t}\mathbb{E}\mathbbm{1}\{\sigma_{s}(x)=\dotsb=\sigma_{s}(y-1)\}\, ds
\]
where in both cases the change of order of integral, summation and
expectation is justified by positivity of the integrands. Hence the
lemma follows from the next claim.\end{proof}
\begin{claim}
\label{claim:Qlonex}$\mathbb{E}Q(x,y,t)\le\E \int_{0}^{t}\mathbbm{1}\{\sigma_{s}(x)=\dotsb=\sigma_{s}(y-1)\}\, ds$.\end{claim}
\begin{proof}
Denote the integrand on the right hand side by $\chi(s)$. By the
definition of a Toom process, $\left.\sigma_{s}\right|_{[x,y]}$ changes
only finitely many times in the time interval $[0,t]$, almost surely.
Hence $\chi$ is Riemann integrable and
\[
\int_{0}^{t}\chi(s)\, ds=\lim_{\epsilon\to0}\epsilon\sum_{i=0}^{\left\lfloor t/\epsilon\right\rfloor }\chi(\epsilon i).
\]

Next let $t_{1}<\dotsb<t_{k}$ be the Poisson arrivals at $x$ in the time
interval $[0,t]$ and let 
\[
B(\epsilon)=\left|\left\{ i:\exists s\in(t_{i}-\epsilon,t_{i})\mbox{ such that }\left.\sigma_{s^{-}}\right|_{[x,y]}\ne\left.\sigma_{s}\right|_{[x,y]}\right\} \right|.
\]
The condition $\left.\sigma_{s}\right|_{[x,y]}$ changes
only finitely many times in the time interval $[0,t]$ also implies
\[
\lim_{\epsilon\to0}B(\epsilon)=0\qquad\mbox{almost surely}.
\]
Hence by dominated convergence ($B$ is bounded by the number
of Poisson arrivals at $x$), $\lim_{\epsilon\to0}\mathbb{E}(B(\epsilon))=0$.
Defining 
\[
\widetilde{Q}=\left|\left\{ i:\chi(\epsilon i )=1 \text{ and } \exists t_{j}\in[\epsilon i,\epsilon(i+1))  \right\} \right|
\]
we get $Q(x,y,t)\le \widetilde{Q}+ B$ (recall that $Q$ requires also the condition $\sigma_s(y)=-\sigma_s(x)$, which is neglected in $\widetilde{Q}$). Since a 
Toom process has the property that $\sigma_{t^{-}}$ is independent of Poisson
arrivals at $[t,\infty)$ we get 
\[
\mathbb{E}\widetilde{Q}=\E\left[\sum_{i=0}^{\left\lfloor t/\epsilon\right\rfloor }\chi(\epsilon i)\mathbbm{1}\{\exists t_{j}\in[\epsilon i,\epsilon(i+1))\}\right]=\E\left[\sum_{i=0}^{\left\lfloor t/\epsilon\right\rfloor }\chi(\epsilon i)\right](\epsilon+O(\epsilon^{2})).
\]
Since $\E(\epsilon\sum_{i=0}^{\left\lfloor t/\epsilon\right\rfloor }\chi(\epsilon i)-  \int_{0}^{t}\chi(s)\, ds)\to 0 $ (again dominated convergence), the claim follows.\end{proof}

\begin{proof}[Proof of \Cref{lem:quenched generator}, averaged version \Cref{eq:annealed generator}]
Let $x<y$ with $y$ in the support of $f$, and let $\epsilon>0$ and assume for convenience also $\epsilon< \frac 14$.
We wish to define the effect of Toom updates from $x$ to $y$ on $f(\sigma_{\epsilon})$.
Therefore let $t_{1},\dotsc,t_{k}$ be the times of the Poisson arrivals
at $x$ in the time interval $[0,\epsilon]$ and define, for $x<y$,
\[
D(x,y,\epsilon)=\sum_{i=1}^{k}(f(\sigma_{t_{i}})-f(\sigma_{t_{i}^{-}}))\mathbbm{1}\{\sigma_{s}(x)=\dotsb=\sigma_{s}(y-1)\neq \sigma_{s}(y)\}.
\]
Since, for all $z$, a Toom process changes at $z$ only finitely many times
in any time interval and $f$ is local, we get $f(\sigma_{\epsilon})-f(\sigma_{0})=\sum_{x<y}D(x,y,\epsilon)$.
From this we conclude
\begin{equation}
\mathbb{E}(f(\sigma_{\epsilon})-f(\sigma_{0}))=\sum_{x<y}\mathbb{E}(D(x,y,\epsilon))\label{eq:f sum D-1}
\end{equation}
where the exchange of sum and expectation is justified as follows:
Let $Q(y,\epsilon)$ be the number of times $\sigma_{y}$ changes sign
in the time interval $[0,\epsilon]$. Then 
\[
\sum_{x<y}D(x,y,\epsilon)\le2||f||_{\infty}\sum_{z \in \supp(f)}Q(z,\epsilon)
\]
which is integrable by Lemma \ref{lem:Ql} and the fact that the sum
over the $z$ is finite. Using dominated convergence gives (\ref{eq:f sum D-1}).

Moving to the behavior as $\epsilon\to0$, let us first show that
\[
\left.\frac{d}{dt}\mathbb{E}(D(x,y,t))\right|_{t=0}=\mathbb{E}(\mathcal{L}_{x,y}f).\qquad\forall x,y
\]
(where the expectation on the right is with respect to $\sigma_{0}$).
To see this let $E(x,y,\epsilon)$ be the event that for some $z\in[x,y]$
and some $t\in[0,\epsilon]$, $\sigma_{t}(z)\ne\sigma_{0}(z)$.
\begin{claim}
\label{claim:twoPoisson}$|\mathbb{E}(D(x,y,\epsilon))-\epsilon\mathbb{E}(\mathcal{L}_{x,y}f)|\le||f||_{\infty}(6\epsilon\mathbb{P}(E)+\epsilon^{2}).$\end{claim}
\begin{proof}
This is claim is justified by playing around with definitions, but let us do it in
detail nonetheless. Let $E_{1}\subset E$ be the event that for some
$z\in[x,y]$  and some $t\in[0,\epsilon]$, $\sigma_{t}(z)\ne\sigma_{0}(z)$,
and in addition there was no Poisson arrival at $x$ during the time
interval $[0,t]$. Then
\begin{equation}\label{eq:E1E}
\mathbb{E}(D\cdot\mathbbm{1}\{E_{1}\})\le2||f||_{\infty}\epsilon\mathbb{P}(E_{1})\le2||f||_{\infty}\epsilon\mathbb{P}(E).
\end{equation}
The first inequality is due to the fact that after $E_{1}$
happens, it is still necessary to have a Poisson arrival in $[t,\epsilon]$,
which has probability less than $\epsilon$, independently of $E_{1}$.
Even if these events both occur, the maximum effect on the value of
$f$ is $2||f||_{\infty}$.

Another case which is easily dispensed with is the event that there are two or more arrivals at $x$ (denote
the number of arrivals by $k$ and this event by $E_{2}$, so $E_{2}=\{k\ge2\}$).
Then $\mathbb{P}(E_{2})<\epsilon\mathbb{P}(E)$ because after the
first arrival (which changes $\sigma_{x}$ hence is included in $E$)
we still need another arrival, independently. Hence $\mathbb{E}(D\cdot\mathbbm{1}\{E_{2}\})<2||f||_{\infty}\epsilon\mathbb{P}(E).$ 

For the remainder (denote it by $E_{3}=E\cap(E_{1}\cup E_{2})^{c}$), let
\[
G=(\mathcal{L}_{x,y}f)(\sigma_{0})\mathbbm{1}\{k=1\}.
\]
Then
\[
\mathbb{E}(D\cdot\mathbbm{1}\{E_{3}\})=\mathbb{E}(G\cdot\mathbbm{1}\{E_{3}\})
\]
because under $E_{3}$ these are exactly the same variables. Hence
\[
|\mathbb{E}(D\cdot\mathbbm{1}\{E_{3}\})-\mathbb{E}(G)|
\le2||f||_{\infty}\mathbb{P}(\{k=1\}\setminus E_{3})
=2||f||_\infty\mathbb{P}(\{k=1\}\cap E_1)
\le2||f||_{\infty}\epsilon\mathbb{P}(E)
\]
where the equality follows because $\{k=1\}\subset E\setminus E_2$, and the last inequality is as in \Cref{eq:E1E}.
Finally the definition of $\mathcal{L}_{x,y}$ gives $\mathbb{E}(G)=\epsilon e^{-\epsilon}\mathbb{E}\mathcal{L}_{x,y}f$
so $|\mathbb{E}(G)-\epsilon\mathbb{E}\mathcal{L}_{x,y}f|\le\epsilon^{2}$
(recall that we assumed $\epsilon<\frac{1}{4}$). Putting everything
together the claim is proved.
\end{proof}
Differentiability of $\mathbb{E}D(x,y,t)$ is now immediate; we
write 
\[
E=\Big\{\sum_{z=x}^{y-1}Q(z,\epsilon)>0\Big\}
\]
and get from  \Cref{lem:Ql} and Markov's inequality that $\mathbb{P}(E)\le\epsilon C(x,y)$.
Claim \ref{claim:twoPoisson} now gives that $\frac{d}{dt}\mathbb{E}D(x,y,t)=\mathbb{E}\mathcal{L}_{x,y}f$.

To be able to sum the derivatives over $x$ 
we use the assumption that $\mathbbm E(l_x(\sigma_0)^{1+\eta})<\infty$ from the local \conda and Markov's inequality.
We conclude that for any $\delta>0$
and $y\in\mathbb{Z}$ there exists $N=N(y)$ such that 
\begin{equation}\label{eq:stationary_stronger}
\mathbb{E}(l_{y}(\sigma_{t})\mathbbm{1}\{l_{y}(\sigma_{t})>N\})\le\delta\qquad\forall t\in[0,t_{0}].
\end{equation}
Assume also $N>\frac{1}{\delta}$. Applying claim \ref{claim:Qlonex}
to all $x<y-N$ and summing gives
\[
\sum_{x<y-N}\mathbb{E}D(x,y,\epsilon)\le2||f||_{\infty}\int_{0}^{\epsilon}l_{y}(\sigma_{t})\mathbbm{1}\{l_{y}(\sigma_{t})>N\}\, dt\le2||f||_{\infty}\cdot\epsilon\delta
\]
so
\begin{align*}
\varlimsup_{\epsilon\to0}\frac{1}{\epsilon}\sum_{x<y}D(x,y,\epsilon) & \le\sum_{y,x<y-N(y)}\mathbb{E}\mathcal{L}_{x,y}f+2||f||_{\infty}\delta\\
\varliminf_{\epsilon\to0}\frac{1}{\epsilon}\sum_{x<y}D(x,y,\epsilon) & \ge\sum_{y,x\in [y-N(y), y)}\mathbb{E}\mathcal{L}_{x,y}f-2||f||_{\infty}\delta.
\end{align*}
Taking $\delta\to0$ gives the lemma.
\end{proof}
\begin{remark}The proof of \Cref{eq:annealed generator} just given can be strengthened slightly when process is stationary. In this case it is enough to assume $\mathbbm E(l_x(\sigma(0)))<C(x)$ for all $x$ i.e.\ it is not necessary to have $1+\eta$ moments, 1 is enough. The only difference in the proof is the justification of \Cref{eq:stationary_stronger}.
\end{remark}
Let us now give the proof of the conditional version of the last lemma. 

\begin{proof}[Proof of \Cref{lem:quenched generator}, \Cref{eq:quenched generator}]
Let $x<y$ with $y$ in the support of $f$, and let $\epsilon\in(0,\frac14)$.
We wish to define the effect of Toom updates from $x$ to $y$ on $f(\sigma_{\epsilon})$.
Therefore let $t_{1},\dotsc,t_{k}$ be the times of the Poisson arrivals
at $x$ in the time interval $[0,\epsilon]$ and define, for $x<y$,
\[
D(x,y,\epsilon)=\sum_{i=1}^{k}(f(\sigma_{t_{i}})-f(\sigma_{t_{i}^{-}}))\mathbbm{1}\{\sigma_{s}(x)=\dotsb=\sigma_{s}(y-1)\neq \sigma_{s}(y)\}.
\]
Since, for all $z$, a Toom process changes at $z$ only finitely many times
in any time interval and $f$ is local, we get $f(\sigma_{\epsilon})-f(\sigma_{0})=\sum_{x<y}D(x,y,\epsilon)$.
  Let $\xi(\sigma)$
be the event that $\sigma(x)=\sigma(x+1)=\dotsb=\sigma(y-1)$. We
define two ``bad'' events, $B^{1}$ and $B^{2}$ as follows.
\begin{enumerate}
\item For $z\in[x,y]$ let $B^{1}(x,z,y,t)$ be the event and that for some
$s<t$ we have $\sigma_{s}(z)\ne\sigma_{0}(z)$; and that for some
$u\in(s,t)$ there was a Poisson arrival at $x$ and $\xi(\sigma_{u^{-}})$
occurred.
\item Let $B^{2}(x,y,t)$ be the event that $\xi(\sigma_{0})$ occurred,
that for some $s<t$ we have $\sigma_{s}(x)\ne\sigma_{0}(x)$, and
that there is a Poisson arrival at $x$ in the time interval $(s,t)$.
\end{enumerate}
Let
\[
\begin{aligned}
B^{1}(t)&=\bigcup_{\substack{x\le z\le y:\\ y \in \supp(f)}}B^{1}(x,z,y,t)\\
B^{2}(t)&=\bigcup_{\substack{x<y:\\ y \in \supp(f)}}B^{2}(x,y,t)
\end{aligned}
\qquad \qquad B(t)=B^{1}(t)\cup B^{2}(t).
\]

\begin{claim}
\label{claim:twoPoissonQ}Almost surely,
\begin{multline*}
\sup_{t\in\left[0,\epsilon\right]}|\mathbb{E}(f(\sigma_{t})-f(\sigma_{0})\,|\,\sigma_{0})-t(\mathcal{L}f)(\sigma_{0})|\\
\le2||f||_{\infty}(\mathbb{P}(B(\epsilon)\,|\,\sigma_{0})+\epsilon^{2})\Big(1+\sum_{y\in \supp(f)}l_{y}(\sigma_{0})\Big).
\end{multline*}
\end{claim}
\begin{proof}
Fix one $y$ in the support of $f$ and one $x<y$. Define the variable
$G(x,y,t)=(\mathcal{L}_{x,y}f)(\sigma_{0})\cdot\mathbbm{1}\{k\ge1\}$
where $k$ is the number of Poisson arrivals at $x$ in the time interval
$[0,t]$ ($G$ depends on $t$ only via $k$). We claim that if $B(\epsilon)$ doesn't happen then $D(x,y,t)=G(x,y,t)$
for all $t\in[0,\epsilon]$. 

Indeed, if $k=0$ then they are both
$0$, and if $k\ge1$ then $G=(\mathcal{L}_{x,y}f)(\sigma_{0})$ while
$D=\sum_{i}(\mathcal{L}_{x,y}f)(\sigma_{t_{i}^{-}})$, where, as usual,
$t_{i}$ are the Poisson arrivals at $x$ in the time interval $[0,t]$.
If $B^{1}$ doesn't occur, only the first
term in this sum may be non-zero (this first arrival will change the
value of $\sigma(x)$, so any further arrival at $x$, if it contributes
to $D$ it must also trigger $B^{1}(x,x,y,\epsilon)$). Further, the
fact that $B^{1}$ doesn't occur implies that 
\[
(\mathcal{L}_{x,y}f)(\sigma_{t_{1}^{-}})\ne0\implies\left.\sigma_{t_{1}^{-}}\right|_{[x,y]}=\left.\sigma_{0}\right|_{[x,y]}
\]
while the fact that $B^{2}$ doesn't occur implies that 
\[
(\mathcal{L}_{x,y}f)(\sigma_{0})\ne0\implies\left.\sigma_{t_{1}^{-}}\right|_{[x,y]}=\left.\sigma_{0}\right|_{[x,y]}.
\]
Thus if neither occurred, $(\mathcal{L}_{x,y}f)(\sigma_{t_{1}^{-}})=(\mathcal{L}_{x,y}f)(\sigma_{0})$
and $D=G$.

Summing over $x$ and $y$ gives, assuming $B=B(\epsilon)$ doesn't
occur, that 
\[
f(\sigma_{t})-f(\sigma_{0})=G(t):=\sum_{x,y}G(x,y,t)
\]
(note that the sum defining $G$ is in fact finite). Taking conditional
expectation gives
\begin{equation}
\mathbb{E}((f(\sigma_{t})-f(\sigma_{0}))\cdot\mathbbm{1}\{B^{c}\}\,|\,\sigma_{0})=\mathbb{E}(G\cdot\mathbbm{1}\{B^{c}\}\,|\,\sigma_{0})\label{eq:noB}
\end{equation}
One remainder can be estimated simply by
\begin{equation}
\mathbb{E}(|f(\sigma_{t})-f(\sigma_{0})|\cdot\mathbbm{1}\{B\}\,|\,\sigma_{0})\le2||f||_{\infty}\mathbb{P}(B\,|\,\sigma_{0})\label{eq:Bonf}
\end{equation}
while for the other we ignore the condition $\{k\ge1\}$ in the definition
of $G$ and get
\begin{align}
\mathbb{E}(|G|\cdot\mathbbm{1}\{B\}\,|\,\sigma_{0})
&\le\mathbb{E}\Big(\sum_{x,y}|\mathcal{L}_{x,y}f(\sigma_{0})|\cdot\mathbbm{1}\{B\}\,|\,\sigma_{0}\Big)\nonumber\\
&=\sum_{x,y}|\mathcal{L}_{x,y}f(\sigma_{0})|\cdot\mathbb{P}(B\,|\,\sigma_{0})
\le2||f||_{\infty}\;\sum_{\mathclap{y\in\supp (f)}}\;l_{y}(\sigma_{0})\mathbb{P}(B\,|\,\sigma_{0})\label{eq:BonG}
\end{align}
Summing (\ref{eq:noB}), (\ref{eq:Bonf}) and (\ref{eq:BonG}) gives
\[
|\mathbb{E}(f(\sigma_{t})-f(\sigma_{0})\,|\,\sigma_{0})-\mathbb{E}(G\,|\,\sigma_{0})|\le2||f||_{\infty}\mathbb{P}(B\,|\,\sigma_{0})\Big(1+\sum_{y}l_{y}(\sigma_{0})\Big)
\]
Finally a Poisson process calculation gives $\mathbb{E}(G(x,y,t)\,|\,\sigma_{0})=(1-e^{-t})\mathcal{L}_{x,y}f(\sigma_{0})$
so $|\mathbb{E}(G\,|\,\sigma_{0})-t(\mathcal{L}f)(\sigma_{0})|\le2||f||_{\infty}\epsilon^{2}\sum_{y}l_{y}(\sigma_{0})$.
Putting everything together the claim is proved.
\end{proof}
Thus the lemma will be proved 
once we estimate $\mathbb P(B)$. We start with
$B^{1}$.
\begin{claim}
\label{claim:B1}With probability $1$,
\[
\lim_{n\to\infty}2^{n}\mathbb{P}(B^{1}(2^{-n})\,|\,\sigma_{0})=0.
\]
\end{claim}
\begin{proof}
Fix $x\le z\le y$. We note two estimates for $B^{1}(x,z,y,\epsilon)$.
First, ignoring the requirement at $z$ gives
\[
\mathbb{P}(B^{1}(x,z,y,\epsilon))\le\epsilon\sup_{t\in[0,\epsilon]}\mathbb{P}(l_{y}>y-x)\le C(y)\cdot\epsilon(y-x)^{-1-\eta}
\]
where the second inequality is due to our moment assumption on $\sigma$
and Markov's inequality. Second, ignoring the requirement that the
arrival at $x$ actually changes $y$, and using only the fact that
an arrival happened we get from \Cref{lem:Ql}
\begin{equation}
\mathbb{P}(B^{1}(x,z,y,\epsilon))\le C(z)\cdot\epsilon^{2}.\label{eq:B1Cz}
\end{equation}
Summing over all $x$ gives 
\[
\mathbb{P}\Big(\bigcup_{x}B^{1}(x,z,y,\epsilon)\Big)\le\epsilon\sum_{x}\min(C(y)(y-x)^{-1-\eta},C(z)\epsilon)\le C(y,z)\epsilon^{1+\eta/(1+\eta)}
\]
where $C(y,z)$ is some constant which depends on $y$ and $z$ but
not on $\epsilon$. Let us denote $B^{1}(z,y,\epsilon)=\bigcup_{x}B^{1}(x,z,y,\epsilon)$.
Markov's inequality now gives for the conditioned events that
\[
\mathbb{P}(\mathbb{P}(B^{1}(z,y,\epsilon)\,|\,\sigma_{0})>\epsilon^{1+\eta/2(1+\eta)})<C(y,z)\epsilon^{\eta/2(1+\eta)}.
\]
This implies, by the Borel-Cantelli Lemma, that with probability $1$, 
\begin{equation}
\lim_{n\to\infty}2^{n}\mathbb{P}(B^{1}(z,y,2^{-n})\,|\,\sigma_{0})=0\label{eq:B1zy}
\end{equation}
for all $z$ and $y$.

There is another important consequence of the conditioning over $\sigma_{0}$.
Fix some value of $\sigma_{0}$ and let $z^{*}=z^{*}(y)=y-l_{y}(\sigma_{0})-1$.
For any $x\le z^{*}$ we have that 
\[
\bigcup_{z=x}^{y}B^{1}(x,z,y,\epsilon)\subset B^{1}(x,z^{*},y,\epsilon)
\]
because the requirement that at the time $u$ ($u$ from the definition
of $B^{1}$) we have the event that $\xi(\sigma_{u^{-}})$ occurs is 
fulfilled only if $\sigma_{u^{-}}(z^{*})\ne\sigma_{0}(z^{*})$. Thus
\[
\varlimsup_{n\to\infty}2^{n}\mathbb{P}\Big(\bigcup_{x\le z^{*}}\bigcup_{z}B^{1}(x,z,y,2^{-n})\,\Big|\,\sigma_{0}\Big)
\le\varlimsup_{n\to\infty}2^{n}\mathbb{P}(B^{1}(z^{*},y,2^{-n})\,|\,\sigma_{0})=0.
\]
The remaining $x$ (i.e.\ $x>z^{*}$) can be estimated directly by
summing (\ref{eq:B1zy}) over $z$ from $z^{*}+1$ to $y$. We get
\[
2^{n}\mathbb{P}\Big(B^{1}(2^{-n})\,\Big|\,\sigma_{0}\Big)
\le2^{n}\sum_{y\in\supp f}\sum_{z=z^{*}(y)}^{y}\mathbb{P}(B^{1}(z,y,2^{-n})\,|\,\sigma_{0})\to0
\]
proving the claim.
\end{proof}
The estimate of $B^{2}(\epsilon)$ is much simpler and we will not
dignify it with a claim. By \Cref{lem:Ql}, $\Pr(B^2(x,y,\eps))<C(x)\eps^2$ (\Cref{lem:Ql} gives that the probability of a change in $\si(x)$ is smaller than $C(x)\eps$, and after that one still needs another Poisson ring at $x$). Markov's inequality then gives
\[
\Pr(\si_0: \Pr(B^2(x,y,\eps)>\eps^{3/2}\,|\,\si_0))<C(x)\eps^{1/2}
\]
so by Borel-Cantelli,
\[
\lim_{n\to\infty}2^{n}\mathbb{P}(B^2(x,y,2^{-n})\,|\,\sigma_{0})=0
\]
almost surely (in $\si_0$). Now, $B^2(2^{-n})$ is just a finite sum of the $B^2(x,y,2^{-n})$: $y$ ranges over the support of $f$ and $x$ runs from $y-l_y(\si_0)$ to $y$. Hence also
\begin{equation}
\lim_{n\to\infty}2^{n}\mathbb{P}(B^2(2^{-n})\,|\,\sigma_{0})=0\label{eq:B2}
\end{equation}
almost surely.

The lemma is now proved. Applying claim \ref{claim:twoPoissonQ} for
$t\in[2^{-n-1},2^{-n}]$ we get 
\begin{multline*}
\Big|\frac{1}{t}\mathbb{E}(f(\sigma_{t})-f(\sigma_{0})\,|\,\sigma_{0})-\mathcal{L}f(\sigma_{0})\Big|\\
\le2||f||_{\infty}(2^{n+1}\mathbb{P}(B(2^{-n})\,|\,\sigma_{0})+2^{-n+1})\Big(1+\sum_{y}l_{y}(\sigma_{0})\Big).
\end{multline*}
By claim \ref{claim:B1} and (\ref{eq:B2}), $2^{n}\mathbb{P}(B(2^{-n})\,|\,\sigma_{0})\to0$
almost surely, and the other terms on the right hand side are constant
for any fixed $\sigma_{0}$. \end{proof}

\bibliography{ToomBib}{}

\begin{thebibliography}{10}

\bibitem{alexander}
R.~Alexander.
\newblock Time evolution for infinitely many hard spheres.
\newblock {\em Commun. Math. Phys.}, 49(3):217--232, 1976.
\newblock
  \href{http://link.springer.com/article/10.1007/BF01608728}{\nolinkurl{springer.com/BF01608728}}.

\bibitem{ToomTsetlin}
Arvind Ayyer, Anne Schilling, Benjamin Steinberg, and Nicolas~M. Thi{\'e}ry.
\newblock Markov chains, {$\mathcal{R}$}-trivial monoids and representation
  theory.
\newblock {\em Internat. J. Algebra Comput.}, 25(1-2):169--231, 2015.
\newblock
  \href{http://www.worldscientific.com/doi/10.1142/S0218196715400081}{\nolinkurl{worldscientific.com/S0218196715400081}}.

\bibitem{BFLS}
G.T. Barkema, P.L. Ferrari, J.L. Lebowitz, and H.~Spohn.
\newblock Kardar-parisi-zhang universality class and the anchored toom
  interface.
\newblock {\em {Phys. Rev. E}}, {90}({4}):{042116}, {Oct} 2014.
\newblock
  \href{http://link.aps.org/doi/10.1103/PhysRevE.90.042116}{\nolinkurl{aps.org/pre/abstract/10.1103/PhysRevE.90.042116}}.

\bibitem{CKR2}
Nick Crawford and Woijcech De~Roeck.
\newblock Invariance principle for `push' tagged particles for a toom
  interface.
\newblock
  \href{https://arxiv.org/abs/1610.07765}{\nolinkurl{arXiv:1610.07765}}, 2016.

\bibitem{DLSS}
B.~Derrida, J.~L. Lebowitz, E.~R Speer, and H.~Spohn.
\newblock Dynamics of an anchored {T}oom interface.
\newblock {\em Journal of Physics A: Mathematical and General}, 24(20):4805,
  1991.
\newblock
  \href{http://iopscience.iop.org/0305-4470/24/20/015}{\nolinkurl{iop.org/0305-4470}}.

\bibitem{DS}
P.~Devillard and H.~Spohn.
\newblock Universality class of interface growth with reflection symmetry.
\newblock {\em Journal of Statistical Physics}, 66(3-4):1089--1099, 1992.
\newblock
  \href{http://dx.doi.org/10.1007/BF01055718}{\nolinkurl{springer.com}}.

\bibitem{krugspohn}
J~Krug and H~Spohn.
\newblock Kinetic roughening of growing surfaces.
\newblock In C.~Godr{\`e}che, editor, {\em Solids far from equilibrium, Vol.
  1}. Cambridge University Press, Cambridge, 1991.

\bibitem{Liggettnonfeller}
T.~M. Liggett.
\newblock Long range exclusion processes.
\newblock {\em Ann. Probab.}, 8(5):861--889, 1980.
\newblock
  \href{http://www.jstor.org/stable/2242933}{\nolinkurl{jstor.org/2242933}}.

\bibitem{LiggettBook}
T.~M. Liggett.
\newblock {\em Interacting Particle Systems}, volume 276.
\newblock Springer, 1985.

\bibitem{LW}
L{\'a}szl{\'o} Lov{\'a}sz and Peter Winkler.
\newblock Mixing times.
\newblock In {\em Microsurveys in discrete probability ({P}rinceton, {NJ},
  1997)}, volume~41 of {\em DIMACS Ser. Discrete Math. Theoret. Comput. Sci.},
  pages 85--133. Amer. Math. Soc., Providence, RI, 1998.
\newblock
  \href{http://citeseerx.ist.psu.edu/viewdoc/summary?doi=10.1.1.53.1717}{\nolinkurl{citeseerx}}.

\bibitem{maesredig}
C.~Maes, F.~Redig, E.~Saada, and A.~Van~Moffaert.
\newblock On the thermodynamic limit for a one-dimensional sandpile process.
\newblock {\em Markov processes and related fields}, 6:1681--1698, 1998.
\newblock \href{http://arxiv.org/abs/math/9810093}{\nolinkurl{arXiv:9810093}}.

\bibitem{PBMM92}
M.~Paczuski, M.~Barma, S.~N. Majumdar, and T.~Hwa.
\newblock Fluctuations of a nonequililbrium interface.
\newblock {\em Phys. Rev. Lett.}, 69:2735--2735, Nov 1992.
\newblock
  \href{http://link.aps.org/doi/10.1103/PhysRevLett.69.2735}{\nolinkurl{aps.org/10.1103}}.

\bibitem{Toomy}
A.~L. Toom.
\newblock Stable and attractive trajectories in multicomponent systems.
\newblock In {\em Multicomponent random systems}, volume~6 of {\em Adv. Probab.
  Related Topics}, pages 549--575. Dekker, New York, 1980.

\end{thebibliography}
\bibliographystyle{plain}

\flushleft
\begin{tabular}{lcr}
\begin{tabular}{l}
Nicholas Crawford \\
Dept. of Mathematics\\
The Technion \\
{\tt nickc@tx.technion.ac.il}
\end{tabular}
&
\begin{tabular}{l}
Gady Kozma\\
Dept. of Math and CS\\
The Weizmann Institute\\
{\tt gady.kozma@weizmann.ac.il}
\end{tabular}
&
\begin{tabular}{l}
Wojciech De Roeck \\
Dept. of Physics\\ 
KU Leuven \\
{\tt wmderoeck@gmail.com}
\end{tabular}
\end{tabular}

\end{document}